\documentclass[journal]{IEEEtran}
\usepackage[T1]{fontenc}
\usepackage[english]{babel}
\usepackage[utf8]{inputenc}
\usepackage{amsthm,amsfonts,amsmath,amssymb}
\usepackage{algorithm}
\usepackage{pifont}
\usepackage{graphicx}
\usepackage{psfrag}

\usepackage{enumitem}
\usepackage{booktabs}
\usepackage{xcolor}
\usepackage{multirow}
\usepackage{array}
\usepackage{steinmetz}
\usepackage{bm}
\usepackage{cite}
\usepackage{subfigure}
\usepackage{makecell}
\newtheorem{prop}{Proposition}
\newtheorem{corollary}{Corollary}
\newtheorem{theorem}{Theorem}

\newtheorem{remark}{Remark}
\newtheorem{assumption}{Assumption}
\usepackage{algorithm}
\usepackage{algpseudocode}
\usepackage[normalem]{ulem}
\usepackage{threeparttable}
\usepackage{siunitx}
\usepackage{empheq}
\usepackage{mathtools}
\usepackage{url}
\usepackage[margin=0.52in]{geometry}
\usepackage{todonotes}
\usepackage{comment}

\ifCLASSINFOpdf

\else

\fi

\allowdisplaybreaks

\begin{document}

\title{Robust Operation of Distribution Networks: Generalized Uncertainty Modelling in Confidence-Level-Based Information Gap Decision}

\author{Zhisheng Xiong,~\IEEEmembership{Graduate Student Member, IEEE, }
        Dimitris Boskos,~\IEEEmembership{Member, IEEE, }
        Bo Zeng,~\IEEEmembership{Member, IEEE, }
        Peter Palensky,~\IEEEmembership{Senior Member, IEEE, }
        Pedro~P.~Vergara,~\IEEEmembership{Senior Member, IEEE} 
\thanks{Zhisheng Xiong, Peter Palensky and Pedro P. Vergara are with the Intelligent Electrical Power Grids (IEPG) group, Delft University of Technology, 2628 CD Delft, The Netherlands (e-mail: \{z.xiong, p.palensky, p.p.vergarabarrios\}@tudelft.nl).

Dimitris Boskos is with the Delft Center for Systems and Control, Delft University of Technology, 2628 CD Delft, The Netherlands (e-mail: d.boskos@tudelft.nl).

Bo Zeng is with the Department of Industrial Engineering and the Department of Electrical and Computer Engineering, University of Pittsburgh, Pittsburgh, PA 15106 USA (e-mail: bzeng@pitt.edu).}}

\maketitle
\begin{abstract}
This paper studies the robust optimal operation of distribution networks (DNs) under renewable generation and load demand uncertainties, seeking an improved trade-off between robustness and economic performance. Building upon information gap decision theory (IGDT), a generalized uncertainty modelling is proposed to enhance the expressiveness of the uncertainty characterization. The proposed modelling captures both symmetric and asymmetric uncertainty features, and supports linear or nonlinear expansion of the uncertainty sets driven by confidence level. This advancement leads to the development of a confidence-level-based IGDT (CL-IGDT) framework for DN operation. To solve the resulting model, its equivalence to a family of two-stage robust optimization problems (TSROs) is established, enabling a Fibonacci search over the confidence level. To further improve computational efficiency, a cut-recycling strategy is proposed to exploit invariant information across TSROs. These techniques are integrated into a novel Fibonacci-Parametric Column-and-Constraint Generation algorithm with guaranteed asymptotic convergence. Case studies validate the effectiveness of the proposed framework and demonstrate the performance advantages of the proposed algorithm.
\end{abstract}


\vspace{-4mm}
\section*{Nomenclature} \vspace{-1mm}
\noindent \emph{A. Sets}:
\begin{description}[itemsep=0.5ex, labelwidth=1.6cm, leftmargin=1.8cm]
	\item [{\small${\cal N}$, ${\cal L}$}]      Set of nodes and lines
  	\item[{\small${\cal G}$, ${\cal B}$}] Set of nodes connecting DGs/ESSs, ${\cal G}, {\cal B} \subset {\cal N}$
	\item [{\small${\cal T}$}]      Set of time-periods
\end{description}
\noindent \emph{B. Indexes}:
\begin{description}[itemsep=0.5ex, labelwidth=1.6cm, leftmargin=1.0cm]
	\item [{\small $i$, $j$, $ij$}]    Node and line indices, $i,j \in {\cal N}$, $ij \in {\cal L}$
	\item [{\small$t$}]           	Time interval $t \in {\cal T}$
\end{description}
\noindent \emph{C. Parameters}:
\begin{description}[itemsep=0.5ex, labelwidth=1.6cm, leftmargin=1.8cm]
	\item [{\small$c_{i}^{G}$, $c_{i}^{U}$}]    			        Operational/start-up cost of DGs
  	\item [{\small$c_{t}^{I}$, $c_{t}^{E}$}]    		            Grid electricity purchase/selling cost
    \item [{\small$\hat{c}^{cur}$, $\hat{c}^{ls}$}]    		        Penalty for renewable curtailment/load shedding
	\item [{\small$\overline{e}_{i,t}$, $\underline{e}_{i,t}$}]     ESS energy limits (max/min)
	\item [{\small$\overline{p}_{i}^{B+}$, $\overline{p}_{i}^{B-}$}]  ESS charging/discharging limits
  	\item [{\small$\eta_{i}^B$}]    				                    ESS efficiency
	\item [{\small$\overline{p}_{i}^{G}$, $\underline{p}_{i}^{G}$}]    				DG active power limits (max/min)
 	\item [{\small$\overline{q}_{i}^{G}$, $\underline{q}_{i}^{G}$}]    				DG reactive power limits (max/min)
    \item[{\small$\kappa^I$, $\kappa^E$}]                   Deviation limits for grid import/export
    \item [{\small$\tilde{p}_{i,t}^{PV}$, $\tilde{p}_{i,t}^{L}$}]    	Expected PV generation/load demand
	\item [{\small$\overline{p}^{l}$}]    				    Line power flow limit
 	\item [{\small$\overline{v}$, $\underline{v}$}]    	    Limit of {\small$\hat{v}_{i,t}$} (max/min)
  	\item [{\small$r_{ij}$}, $x_{ij}$]    	                Line resistance/reactance
 	\item [{\small$RU_{i}^{G}$, $RD_{i}^{G}$}]              DG ramp up/down limits
	\item [{\small$\Delta t$}]    				            Time step length
    \item [{\small$\delta_{i,t}$}]    		                Power factor angle
    \item[{\small$|\cdot|$}]                                Cardinality of a set

\end{description}
\noindent \emph{D. Continuous Variables}:
\begin{description}[itemsep=0.5ex, labelwidth=1.6cm, leftmargin=1.8cm]
        \item [{\small$p_{i,t}^{G}$, $\hat{q}_{i,t}^{G}$}]     DG active/reactive power output
		\item [{\small$e_{i,t}$}]    					       Energy of the ESSs
        \item [{\small$p_{i,t}^{B+}$, $p_{i,t}^{B-}$}]    	   ESS charging/discharging power
        \item [{\small$p_{t}^{I}$, $p_{t}^{E}$}]               Power imported/exported from grid
        \item [{\small$\hat{p}_{i,t}^{cur}$, $\hat{p}_{i,t}^{ls}$}]           Renewable curtailment and load shedding power
        \item [{\small$p_{ij,t}$, $\hat{q}_{ij,t}$}]        Line active/reactive power flow
        \item [{\small$\hat{v}_{i,t}$}]    		Squared bus voltage magnitude
        \item [{\small$\hat{\cdot}$}]    	    Second-stage counterpart
\end{description}
\noindent \emph{E. Integer Variables}:
\begin{description}[itemsep=0.5ex, labelwidth=1.6cm, leftmargin=1.8cm]
    \item[{\small$z_{i,t}^{G}$}]                 DG ON/OFF status (binary)
    \item[{\small$z_{i,t}^{B}$}]                 ESS charging/discharging status (binary)
    \item[{\small$z_t^l$}]                       Grid import/export status (binary)
    \item[{\small$s_{i,t}^{G}$}]                 DG startup count (integer)
\end{description}

\IEEEpeerreviewmaketitle

\section{Introduction} \label{sec:intro}
\IEEEPARstart{T}{he} global transition toward carbon neutrality has led many countries to commit to net-zero emission targets within the coming decades~\cite{carbon_neutrality}. The power sector is undergoing a fundamental transformation, with increasing emphasis on integrating renewable energy sources (RESs) into distribution networks (DNs). This shift is critical for climate goals but also introduces considerable operational risk due to the intermittent and stochastic nature of RESs. In particular, uncertainties in RESs and load demand can result in higher costs, increased power loss, and unstable operation~\cite{impact_uncertainty}. To mitigate such risks, uncertainty-aware optimization models have become essential in modern DN operations, with their effectiveness largely determined by how uncertainties are represented and integrated.

Recent years have seen extensive research on uncertainty-aware operation of DNs, focusing on modeling and handling uncertainties in renewable generation and load demand. A two-stage stochastic programming (SP) model is adopted in~\cite{sp1}, where the first stage minimizes investment cost and the second minimizes expected recourse cost. Ref.~\cite{sp2} applies chance constraints with Gaussian-based scenarios for Volt/Var control. In~\cite{ro1}, a two-stage robust optimization (RO) framework manages active power injections under wind and load uncertainty. Ref.~\cite{ro2} introduces a budgeted uncertainty set in a three-stage model to reduce energy loss and voltage magnitude deviation. In reference to distributionally robust optimization (DRO), Ref.~\cite{dro1} uses a Wasserstein-based ambiguity set for unit commitment, while Ref.~\cite{dro2} models power delivery deviations as distributionally robust chance constraints. Despite differing philosophies, these approaches share a common assumption: uncertainties follow a predefined paradigm, i.e., a probability distribution in SP, an uncertainty set in RO, or an ambiguity set in DRO. Such modeling is termed \textit{decision-independent uncertainty (DIU)}~\cite{DDUs}, or exogenous uncertainty, as it remains unaffected by decisions.

An alternative line of research focuses on information gap decision theory (IGDT), which offers a fundamentally different philosophy for optimization under uncertainty. Instead of minimizing worst-case cost for a given uncertainty characterization, IGDT maximizes the tolerable uncertainty horizon under a predefined performance requirement~\cite{igdt_Ben_Haim}. For example, within a fixed economic budget, Ref.~\cite{w_igdt} develops a weighted IGDT framework for optimal power flow, where the uncertainty horizon expands with wind power penetration rate, capturing time-varying characteristics. Building on the envelope-bound uncertainty set of IGDT, Ref.~\cite{idgt_binary1} introduces binary variables to refine structural definitions of uncertainty sets, enabling a multi-energy dispatch model that maximizes risk-tolerant boundaries. Another structural variant is explored in~\cite{idgt_binary2}, where polyhedral uncertainty sets enlarge the tolerable region for wind power and load demand. Collectively, these approaches have expanded the modeling scope of IGDT. However, they still rely on symmetric uncertainty sets that scale only linearly with financial budgets, which limits adaptability and prevents them from capturing asymmetric features of uncertainty. In DNs, asymmetric uncertainties typically manifest as scenarios of low renewable generation combined with high load demand, which significantly deviate from expected values. A symmetric set may fail to capture these cases and pose operational risks. In our previous work~\cite{zhishengxiong}, a confidence-level-driven uncertainty set within a CL-IGDT framework is proposed, explicitly addressing such asymmetry. However, this set fails to represent structured dependencies among uncertainty dimensions and its box-type expansion still limits diverse uncertainty patterns.

The uncertainty modeling in IGDT is closely tied to performance budget: the uncertainty set evolves with the associated decision that reflects the predefined performance requirement. This feature corresponds to \textit{decision-dependent uncertainty (DDU)}, or endogenous uncertainty. Although not explicitly recognized in IGDT studies, this perspective opens the possibility of accommodating more expressive uncertainty representations. Guided by this insight, we introduce a generalized uncertainty set under the CL-IGDT framework. Conventional IGDT-based uncertainty modeling typically restricts symmetric sets to linear expansions of box-type shapes. In contrast, the proposed set captures both symmetric and asymmetric features of uncertainty, and supports linear or nonlinear expansions with respect to the confidence level. Moreover, compared with prior studies that rely on binary variables, we introduce a continuous structure parameter to represent dependencies among uncertainty dimensions (e.g., temporal correlations). These extensions enhance the modeling capability of IGDT-based approaches. In this way, our formulation encompasses many existing uncertainty models in IGDT-based energy system optimization, such as those in~\cite{w_igdt, idgt_binary1, idgt_binary2}, while overcoming their limited adaptability. From an operational perspective, it enables a more flexible characterization of uncertainties and supports an improved trade-off between robustness and economic performance.

Nonetheless, one must acknowledge the trade-off between modeling sophistication and computational efficiency. As the enhanced uncertainty modeling is embedded into the CL-IGDT framework, solving the resulting optimization problem becomes increasingly challenging. To address this, we establish a mathematical equivalence between CL-IGDT and a family of TSRO problems (TSROs). Leveraging this equivalence, \textit{a cut-recycling strategy} is proposed to exploit invariant information across iterations. Accordingly, a Fibonacci-Parametric Column-and-Constraint Generation (F-PC\&CG) algorithm is developed. Overall, the main contributions are as follows:
\begin{itemize}
\item A robust operation model for DNs is formulated under the proposed CL-IGDT framework on top of our previous work~\cite{zhishengxiong}, achieving an improved trade-off between robustness and economic performance.
\item A generalized uncertainty set is proposed within the CL-IGDT framework, enabling a more expressive characterization of RES and load uncertainties, and enhancing the uncertainty modeling capability of IGDT-based methods.
\item A novel F-PC\&CG algorithm is developed to solve the model. It leverages the equivalence between CL-IGDT and TSROs and incorporates a cut-recycling strategy for computational efficiency, guaranteeing asymptotic convergence with bounded approximation error.
\end{itemize}

\section{Problem Formulation}
For DN operation, the goal is to develop a cost-effective control strategy that ensures reliable power supply and system security. This requires coordinating the operation of distributed energy resources (DERs), including distributed generators (DGs) and energy storage systems (ESSs), to meet demand while satisfying power balance, voltage magnitude limits, and line flow constraints. Accordingly, a two-stage operational framework is proposed. The first stage involves approximate active power dispatch based on expected uncertainty realizations, while the second stage determines control actions under actual conditions, including both adjustments to baseline active power output and new decisions such as reactive power regulation. By coordinating active and reactive power control across timescales, the proposed strategy enhances operational flexibility and reduces conservatism compared to single-stage approaches~\cite{yangyue}. The model is formulated from a forward-looking perspective, assuming a centralized decision-maker responsible for coordinating DER operations and network security (e.g., a future DSO with extended responsibilities). The proposed framework is applicable to operators facing similar uncertainty-aware operational challenges.
\vspace{-10pt}
\subsection{Two-Stage Optimal Operation of DNs}

\subsubsection{First Stage: Baseline Dispatch Planning}
The first-stage approximate active power schedule serves as a baseline that anticipates uncertainties in the second stage, ensuring that subsequent recourse actions remain feasible and cost-effective under various realizations. The corresponding formulation is given in \eqref{eq:obj1}--\eqref{eq:power_limit}. The objective function in~\eqref{eq:obj1} minimizes the total operational cost over the planning horizon~$\mathcal{T}$, including generation cost of DGs, start-up costs, and power exchange with the upstream grid.
\begin{equation}\label{eq:obj1}
\begin{aligned}
\min \quad & \sum_{t \in {\cal T}} \Bigg( \sum_{i \in {\cal G}} \left( c_{i}^{G} p_{i,t}^G + c_{i}^{U} s_{i,t}^G \right) + c_{t}^{I} p_t^{I} - c_{t}^{E} p_t^{E} \Bigg)
\end{aligned}
\end{equation}
subject to:
\setlength{\arraycolsep}{-0.6em}
\begin{eqnarray}
&&z_{i,t}^{G}\underline{p}_{i}^{G} \leq p_{i,t}^{G} \leq z_{i,t}^{G}\overline{p}_{i}^{G},\forall i\in\mathcal{G},\forall t\in\mathcal{T} \label{eq:DGs_limit} \\
&&p_{i,t}^{G} - p_{i,t-1}^{G} \leq RU_{i}^{G},\forall i\in\mathcal{G},\forall t\in\mathcal{T} \label{eq:DGs_rampup} \\
&&p_{i,t-1}^{G} - p_{i,t}^{G} \leq RD_{i}^{G}, \forall i\in\mathcal{G},\forall t\in\mathcal{T} \label{eq:DGs_rampdown} \\
&&s_{i,t}^{G} \geq z_{i,t}^{G}-z_{i,t-1}^{G}, \forall i\in\mathcal{G},\forall t\in\mathcal{T} \label{eq:DGs_startup1} \\
&&s_{i,t}^{G} \geq 0, \forall i\in\mathcal{G},\forall t\in\mathcal{T} \label{eq:DGs_startup2} \\
&&e_{i,t} = e_{i,t-1} + \left(p_{i,t}^{B+}\eta_{B} - p_{i,t}^{B-}/\eta_{i}^{B}\right)\Delta t, \forall i\in\mathcal{B},\forall t\in\mathcal{T} \label{eq:ESSs_energy} \\
&&\underline{e}_{i} \leq e_{i,t} \leq \overline{e}_{i},\forall i\in\mathcal{B},\forall t\in\mathcal{T} \label{eq:ESSs_limit} \\
&&0\leq p_{i,t}^{B+} \leq z_{i,t}^{B}\overline{p}_{i}^{B+},\forall i\in\mathcal{B},\forall t\in\mathcal{T} \label{eq:ESSs_charge} \\
&&0\leq p_{i,t}^{B-} \leq (1-z_{i,t}^{B})\overline{p}_{i}^{B-},\forall i\in\mathcal{B},\forall t\in\mathcal{T} \label{eq:ESSs_discharge} \\
&&0\leq p_{t}^{I} \leq z_{t}^{l}\overline{p}^{l},\forall t\in\mathcal{T} \label{eq:import_limit} \\
&&0\leq p_{t}^{E} \leq (1-z_{t}^{l})\overline{p}^{l},\forall t\in\mathcal{T} \label{eq:export_limit} \\
&&p_{01,t} = p_t^{I} - p_t^{E},\forall t\in\mathcal{T} \label{eq:interface_flow} \\
&&\sum\limits_{j:(ij)\in\mathcal{L}}p_{ij,t}-\sum\limits_{h:(hi)\in\mathcal{L}}p_{hi,t}=p_{i,t},\forall i\in\mathcal{N}\setminus\{0\},\forall t\in\mathcal{T} \label{eq:power_balance} \\
&&\nonumber p_{i,t}=p_{i,t}^{G}+p_{i,t}^{B-}-p_{i,t}^{B+} \\
&&\qquad\qquad\qquad\quad +\tilde{p}_{i,t}^{PV}-\tilde{p}_{i,t}^{L},\forall i\in\mathcal{N}\setminus\{0\},\forall t\in\mathcal{T} \label{eq:netload} \\
&&-\overline{p}^{l}\leq p_{ij,t} \leq \overline{p}^{l},\forall(i,j) \in\mathcal{L},\forall t\in\mathcal{T} \label{eq:power_limit}
\end{eqnarray}

In the above formulation, the DGs constraints are composed of the active power output limits \eqref{eq:DGs_limit}, ramp up/down rate limits \eqref{eq:DGs_rampup}--\eqref{eq:DGs_rampdown}, and start-up counting \eqref{eq:DGs_startup1}--\eqref{eq:DGs_startup2}. The ESSs are considered to operate in power control mode, allowing pre-scheduling of their active power output \cite{active_power_ESS}. The operation of ESSs is governed by constraints on stored energy dynamics \eqref{eq:ESSs_energy}, storage limits \eqref{eq:ESSs_limit}, and charging/discharging limits \eqref{eq:ESSs_charge}--\eqref{eq:ESSs_discharge}. The electricity exchange constraints are given in \eqref{eq:import_limit}--\eqref{eq:interface_flow}. Constraint \eqref{eq:power_balance} models the active power balance, and constraint \eqref{eq:netload} considers net active loads. The maximum power flow limits are enforced by~\eqref{eq:power_limit}. 

\subsubsection{Second Stage: Recourse Control Strategy}
This stage enables the system to respond to deviations and maintain operational security under uncertainty. Once the uncertain scenarios are realized, corrective control actions are determined in the second stage to ensure secure and cost-effective operation. These include active power adjustments to the baseline plan and new decisions such as reactive power outputs from DGs. The total recourse control cost in \eqref{eq:obj2} is determined by the cost of recourse active power output of DGs, additional electricity exchange with the upstream grid, and penalty costs of renewable curtailment and load shedding.
\begin{equation} \label{eq:obj2}
\begin{aligned}
\min \quad & \sum_{t \in {\cal T}} \Bigg( 
\sum_{i \in {\cal G}} c_{i}^{G} \Delta p_{i,t}^G 
+ \hat{c}_{t}^{I} \Delta p_t^{I} 
- \hat{c}_{t}^{E} \Delta p_t^{E} \\
& \qquad + \sum_{i \in \mathcal{N}} \left( \hat{c}^{\rm cur} \hat{p}_{i,t}^{\rm cur} + \hat{c}^{\rm ls} \hat{p}_{i,t}^{\rm ls} \right) 
\Bigg)
\end{aligned}
\end{equation}
subject to:
\setlength{\arraycolsep}{-0.6em}
\begin{eqnarray}
&&\Delta p_{i,t}^{G} = \hat{p}_{i,t}^{G} - p_{i,t}^{G},\forall i\in\mathcal{G},\forall t\in\mathcal{T} \label{eq:DGs_adjust} \\
&&\Delta p_{t}^{I} = \hat{p}_{t}^{I} - p_{t}^{I},\forall t \in \mathcal{T} \label{eq:import_adjust} \\
&&\Delta p_{t}^{E} = \hat{p}_{t}^{E} - p_{t}^{E},\forall t \in \mathcal{T} \label{eq:export_adjust} \\
&&z_{i,t}^{G}\underline{p}_{i}^{G} \leq \hat{p}_{i,t}^{G} \leq z_{i,t}^{G}\overline{p}_{i}^{G},\forall i\in\mathcal{G},\forall t\in\mathcal{T} \label{eq:DGs_recourse} \\
&&\hat{p}_{i,t}^{G} - \hat{p}_{i,t-1}^{G} \leq RU_{i}^{G},\forall i\in\mathcal{G},\forall t\in\mathcal{T} \label{eq:DGs_regulate_rampup} \\
&&\hat{p}_{i,t-1}^{G} - \hat{p}_{i,t}^{G} \leq RD_{i}^{G}, \forall i\in\mathcal{G},\forall t\in\mathcal{T} \label{eq:DGs_regulate_rampdn} \\
&&z_{i,t}^{G} \underline{q}_{i}^{G} \leq \hat{q}_{i,t}^{G} \leq z_{i,t}^{G} \overline{q}_{i}^{G},\forall i\in\mathcal{G},\forall t\in\mathcal{T} \label{eq:DGs_reactive} \\
&&p_{t}^{I} \leq \hat{p}_{t}^{I} \leq (1+\kappa^I)p_{t}^{I},\forall t\in\mathcal{T} \label{eq:import_limit_2s} \\
&&p_{t}^{E} \leq \hat{p}_{t}^{E} \leq (1+\kappa^E)p_{t}^{E},\forall t\in\mathcal{T} \label{eq:export_limit_2s} \\
&&\hat{p}_{01,t} = \hat{p}_t^{I} - \hat{p}_t^{E},\forall t\in\mathcal{T} \label{eq:interface_flow_2s} \\
&&\sum\limits_{j:(ij)\in\mathcal{L}}\hat{p}_{ij,t} - \sum\limits_{h:(hi)\in\mathcal{L}}\hat{p}_{hi,t} = \hat{p}_{i,t},\forall i\in\mathcal{N}\setminus\{0\},\forall t\in\mathcal{T} \label{eq:power_balance_2s} \\
&&\sum\limits_{j:(ij)\in\mathcal{L}}\hat{q}_{ij ,t} - \sum\limits_{h:(hi)\in\mathcal{L}}\hat{q}_{hi,t} = \hat{q}_{i,t},\forall i\in\mathcal{N}\setminus\{0\},\forall t\in\mathcal{T} \label{eq:reactivepower_balance_2s} \\
&& \nonumber \hat{p}_{i,t} = \hat{p}_{i,t}^{G} + p_{i,t}^{B-} - p_{i,t}^{B+}  + p_{i,t}^{PV} \\
&& \qquad\qquad\quad\quad - p_{i,t}^{L} + \hat{p}_{i,t}^{ls} - \hat{p}_{i,t}^{cur}, \forall i\in\mathcal{N}\setminus\{0\},\forall t\in\mathcal{T} \label{eq:netload_2s}\\
&&\hat{p}_{i,t}^{ls}, \, \hat{p}_{i,t}^{cur} \geq 0, \forall i\in\mathcal{N}\setminus\{0\},\forall t\in\mathcal{T} \label{eq:auxiliary} \\
&&\hat{q}_{i,t}=\hat{q}_{i,t}^{G}-q_{i,t}^{L},\forall i\in\mathcal{N}\setminus\{0\},\forall t\in\mathcal{T} \label{eq:reac_netload_2s} \\
&&q_{i,t}^{L}=p_{i,t}^{L}\tan\delta_{i,t},\forall i\in\mathcal{N}\setminus\{0\},\forall t\in\mathcal{T} \label{eq:reac_load_2s} \\
&&\hat{v}_{i,t}-\hat{v}_{j,t}=\frac{r_{ij}\hat{p}_{ij,t}+x_{ij}\hat{q}_{ij,t}}{1-\varphi_{ij}},\forall (i,j)\in\mathcal{L},\forall t\in\mathcal{T} \label{eq:voltage_2s} \\
&&-\overline{p}^{l} \leq \hat{p}_{ij,t} \leq \overline{p}^{l},\forall (i,j) \in\mathcal{L},\forall t\in\mathcal{T} \label{eq:power_limit_2s} \\
&&\hat{v}_{0,t} = 1,\forall t\in\mathcal{T} \label{eq:voltage0_limit_2s} \\
&&\underline{v} \leq \hat{v}_{i,t} \leq \overline{v},\forall i\in\mathcal{N}\setminus\{0\},\forall t\in\mathcal{T} \label{eq:voltage_limit_2s}
\end{eqnarray}

The adjustment variables defined in \eqref{eq:DGs_adjust}--\eqref{eq:export_adjust} capture deviations from the first-stage dispatch. The recourse constraints of DGs are given in \eqref{eq:DGs_recourse}--\eqref{eq:DGs_regulate_rampdn}, together with reactive power output limits in \eqref{eq:DGs_reactive}. The electricity exchange in the second stage is governed by \eqref{eq:import_limit_2s}--\eqref{eq:interface_flow_2s}. Constraints \eqref{eq:power_balance_2s}-\eqref{eq:reactivepower_balance_2s} model the active/reactive power balance equations. Constraints \eqref{eq:netload_2s}-\eqref{eq:reac_load_2s} represent net active/reactive loads. The lossy DistFlow formulation \cite{lossy_DistFlow} is employed to model voltage magnitude drop in lines in \eqref{eq:voltage_2s}. The parameter $\varphi_{ij}$ approximates loss terms, linearizing voltage calculations while preserving accuracy. The maximum power flow and voltage magnitude limits are enforced by constraints~\eqref{eq:power_limit_2s} and~\eqref{eq:voltage0_limit_2s}-\eqref{eq:voltage_limit_2s}, respectively.

\section{Confidence-Level-Based IGDT Framework}
To tackle uncertainties in renewable generation and load demand, the DN operation problem is formulated within the IGDT framework. IGDT provides a generic structure composed of three elements: system model, uncertainty modelling, and performance requirement~\cite{igdt_Ben_Haim}. Within this framework, robustness is evaluated by identifying the maximum uncertainty level while still meeting a predefined performance requirement. Building on this foundation, we reformulate the two-stage DN optimal operation model under an enhanced CL-IGDT framework.
\vspace{-10pt}
\subsection{Compact Formulation of CL-IGDT}

In our previous work~\cite{zhishengxiong}, IGDT was extended to CL-IGDT by enhancing uncertainty modeling with confidence-level-driven uncertainty sets and switching the objective function from the robustness measure to a confidence-level-inspired robustness measure $\alpha$. Robustness is quantified through uncertainty sets that capture asymmetric features and represent multiple uncertainty sources in a decoupled manner through a unified $\alpha$. Building on this, we further enhance the uncertainty modelling by introducing a generalized uncertainty set. For clarity, we first express the CL-IGDT framework in compact form, referred to as \textbf{CL-IGDT (functional)}. $\mathcal{U}(\alpha)$ remains abstract and will be specified in the next subsection.

\vspace{1.5mm}
\!\textbf{CL-IGDT (functional):}
\vspace{-1mm}
\begin{align}
\alpha^*(&\Lambda) = \max_{\substack{\alpha\in[0,1],\\ \bm{x},\, \bm{y}(\cdot)}} ~\alpha \label{eq:cldigdt-obj} \\
\text{s.t.}~  &\bm{A} \bm{x} \leq \bm{b}, \label{eq:cldigdt-con1} \\
& \bm{B}_1 \bm{x} + \bm{B}_2 \bm{y}(\bm{u}) \geq \bm{d}- \bm{E} \bm{u}, \forall \bm{u} \in \mathcal{U}(\alpha), \label{eq:cldigdt-con2} \\
& \bm{c}_1^\top \bm{x} + \max_{\bm{u} \in \mathcal{U}(\alpha)} \bm{c}_2^\top \bm{y}(\bm{u}) \leq \Lambda, \label{eq:cldigdt-con3} \\
& \bm{x}\in\left\lbrace0,1 \right\rbrace^{n_{\rm x_1}}\times \mathbb{R}_+^{n_{\rm x_2}}, \, \bm{y}(\bm{u}) \in\mathbb{R}_+^{n_{\rm y}}, \label{eq:cldigdt-con4}
\end{align}
where
\begin{eqnarray}
&&\mathcal{U}(\alpha)=\left\{\bm{u}\in\mathbb{R}_{+}^{n_{\rm u}}\;\big|\; \bm{F}(\alpha) \bm{u} \leq \bm0\right\}. \label{eq:Uset}
\end{eqnarray}

In this formulation, $\bm{x}$ denotes the first-stage decision vector, including $p_{i,t}^{G}$, $p_{t}^{I}$, $p_{t}^{E}$, $p_{i,t}^{B-}$, $p_{i,t}^{B+}$, $z^G_{i,t}$, $z^B_{i,t}$ and $z^l_{t}$. The second-stage decision vector $\bm{y}$ consists of $\hat{p}_{i,t}^{G}$, $\hat{q}_{i,t}^{G}$, $\hat{p}_{t}^{I}$, $\hat{p}_{t}^{E}$, $\hat{p}_{i,t}^{ls}$ and $\hat{p}_{i,t}^{cur}$. The uncertainty vector $\bm{u}$ represents renewable generation and load demand realizations. The integers ${n_{\rm x_1}}/{n_{\rm x_2}}$, ${n_{\rm u}}$ and ${n_{\rm y}}$ denote the dimensions of $\bm{x}$, $\bm{u}$ and $\bm{y}$, respectively. The coefficient vectors $\bm{b}$, $\bm{c_1}$, $\bm{c_2}$, $\bm{d}$ and constraint matrices $\bm{A}$, $\bm{B_1}$, $\bm{B_2}$, $\bm{E}$, $\bm{F}$ are of appropriate dimensions, while ${m_{\rm u}}$ and ${m_{\rm y}}$ denote the numbers of rows of $\bm{F}$ and $\bm{B}_2$, respectively. Constraints~\eqref{eq:cldigdt-con1} compactly represent \eqref{eq:DGs_limit}-\eqref{eq:power_limit}. Constraint~\eqref{eq:cldigdt-con2} captures the coupling between first- and second-stage decisions, corresponding to \eqref{eq:DGs_adjust}-\eqref{eq:voltage_limit_2s}. The costs in~\eqref{eq:obj1} and~\eqref{eq:obj2} are integrated into the constraint~\eqref{eq:cldigdt-con3}, which enforces the predefined financial budget~$\Lambda$.
\vspace{-6pt}
\subsection{Uncertainty Modeling} \label{subsec:uncertainty_modeling}
\subsubsection{Confidence-Level-Driven Uncertainty Set} Following~\cite{zhishengxiong}, we consider a confidence-level-driven uncertainty set constructed from confidence intervals derived from distributional information. Specifically, the uncertainty set is defined by
\begin{equation}\label{eq:idm_uncertainty set}
\mathcal{U}_{CL}(\alpha) := [\underline{\bm{u}}^{\dagger}(\alpha), \overline{\bm{u}}^{\dagger}(\alpha)]
\end{equation}
for some
\begin{equation}
\begin{aligned}
(\underline{\bm{u}}^{\dagger}(\alpha), \overline{\bm{u}}^{\dagger}(\alpha)) \in &\arg\min_{(\underline{\bm{u}},\overline{\bm{u}})} \{\|\overline{\bm{u}} - \underline{\bm{u}}\|_{1} \,| \\
& \underline{F}_{u_t}(\overline{u}_{t}) - \overline{F}_{u_t}(\underline{u}_{t}) \geq \alpha, \forall t \in \mathcal{T}\}.
\end{aligned}
\end{equation}
where the parameter $\alpha$ here serves as a decision variable in CL-IGDT, implying that the resulting set exhibits DDU property. For each component $u_t$ of the uncertainty, $\underline{F}_{u_t}$ and $\overline{F}_{u_t}$ denote the lower- and upper-bound CDFs associated with an ambiguity set $\mathcal{A}_t$ of plausible CDFs. In particular, 
\begin{equation}
\underline{F}_{u_t}(u):=\min\{F_{u_t}(u) \,|\, F_{u_t}\in\mathcal A_t\},
\end{equation}
and likewise for $\overline{F}_{u_t}(\underline{u}_{t})$. The ambiguity sets can be constructed in a data-driven manner following \cite{idm_Li_Peng}, thereby avoiding reliance on a precise probability distribution.

As illustrated in Fig.~\ref{fig:idm_uncertainty set}, this uncertainty modeling is well suited to asymmetric uncertainties. Owing to its confidence-driven interpretation, the uncertainty set typically grows nonlinearly with respect to $\alpha$. This differs from the IGDT framework, where a horizon-scaling variable controls the size of the uncertainty set, yielding linear and uniform expansion across all dimensions.

\begin{figure}[!t]
\centering
\includegraphics[width=2.6in]{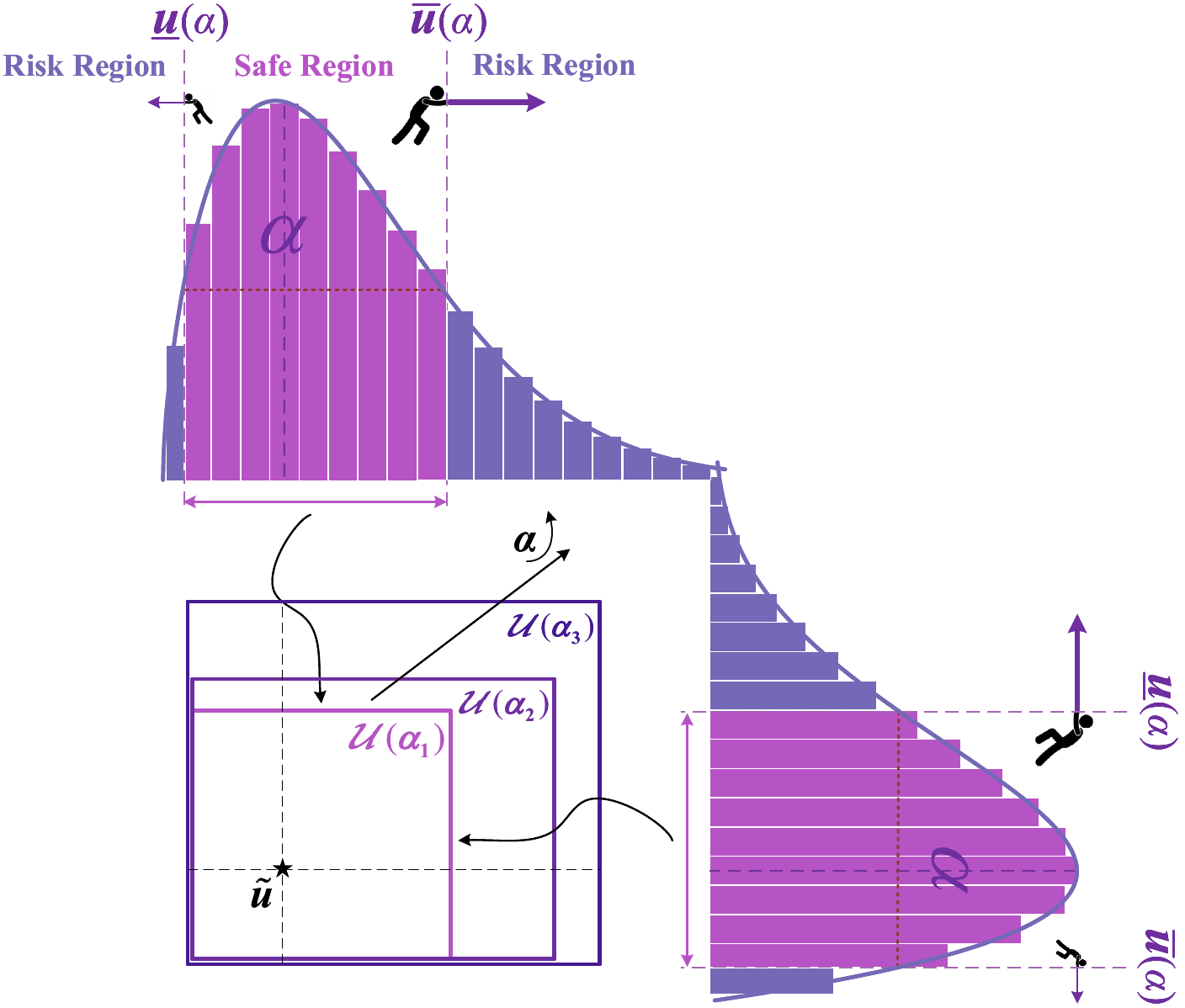}
 \caption{Illustration of a two-dimensional confidence-level-driven uncertainty set. The essence of $\mathcal{U}(\alpha)$ can be interpreted as a box uncertainty set, which is constructed from confidence intervals parameterized by $\alpha$.}
\label{fig:idm_uncertainty set}
\end{figure}

\subsubsection{Generalized Uncertainty Set} 
To obtain a tighter uncertainty representation beyond component-wise bounds, we draw inspiration from the budget uncertainty set~\cite{RO} and introduce a structural parameter~$\Gamma$ that couples deviations across dimensions to refine $\mathcal U_{CL}(\alpha)$. This yields the generalized uncertainty set
\begin{equation}
\label{eq:generalized_uncertainty_set}
\begin{aligned}
\mathcal{U}&(\alpha) 
:=  \{\bm{u} - \bm{\epsilon}^{-} \circ (\bm{u} - \underline{\bm{u}}^{\dagger}(\alpha)) 
+ \bm{\epsilon}^{+} \circ (\overline{\bm{u}}^{\dagger}(\alpha) - \bm{u})\,| \\ 
&\;\,\bm{u} \in \mathcal{U}_{CL}(\alpha),\frac{\|\bm{\epsilon}^{-} + \bm{\epsilon}^{+}\|_1}{|\mathcal{T}|} \leq \Gamma, \; \bm{0} \leq \bm{\epsilon}^{-}, \bm{\epsilon}^{+} \leq \bm{1} \Big\},
\end{aligned}
\end{equation}
where $\circ$ denotes Hadamard product. This formulation can be expressed in the compact form of~\eqref{eq:Uset}.

The proposed formulation provides explicit control of $\mathcal{U}(\alpha)$, with $\alpha$ governing its size and $\Gamma$ its structure, as illustrated in Fig.~\ref{fig:generalized Uset}. Specifically, the proposed uncertainty set can expand with respect to $\alpha$ in a linear or nonlinear, symmetric or asymmetric manner, depending on the distributional characteristics. Moreover, $\alpha$ controls all components of $\bm{u}$, while their uncertainty bounds are determined by the corresponding distributions (e.g., those of solar power and load demand), and $\Gamma$ captures dimensional correlations and structural coupling.

Thanks to this expressiveness, $\mathcal{U}(\alpha)$ encompasses most existing uncertainty modeling constructions in IGDT-based energy system optimization as special cases, as discussed in Section~\ref{sec:intro}. This makes the CL-IGDT framework applicable to diverse uncertainty profiles in DNs and supports a broader range of operational decision preferences.

\begin{figure}[!t]
\centering
\includegraphics[width=3.2in]{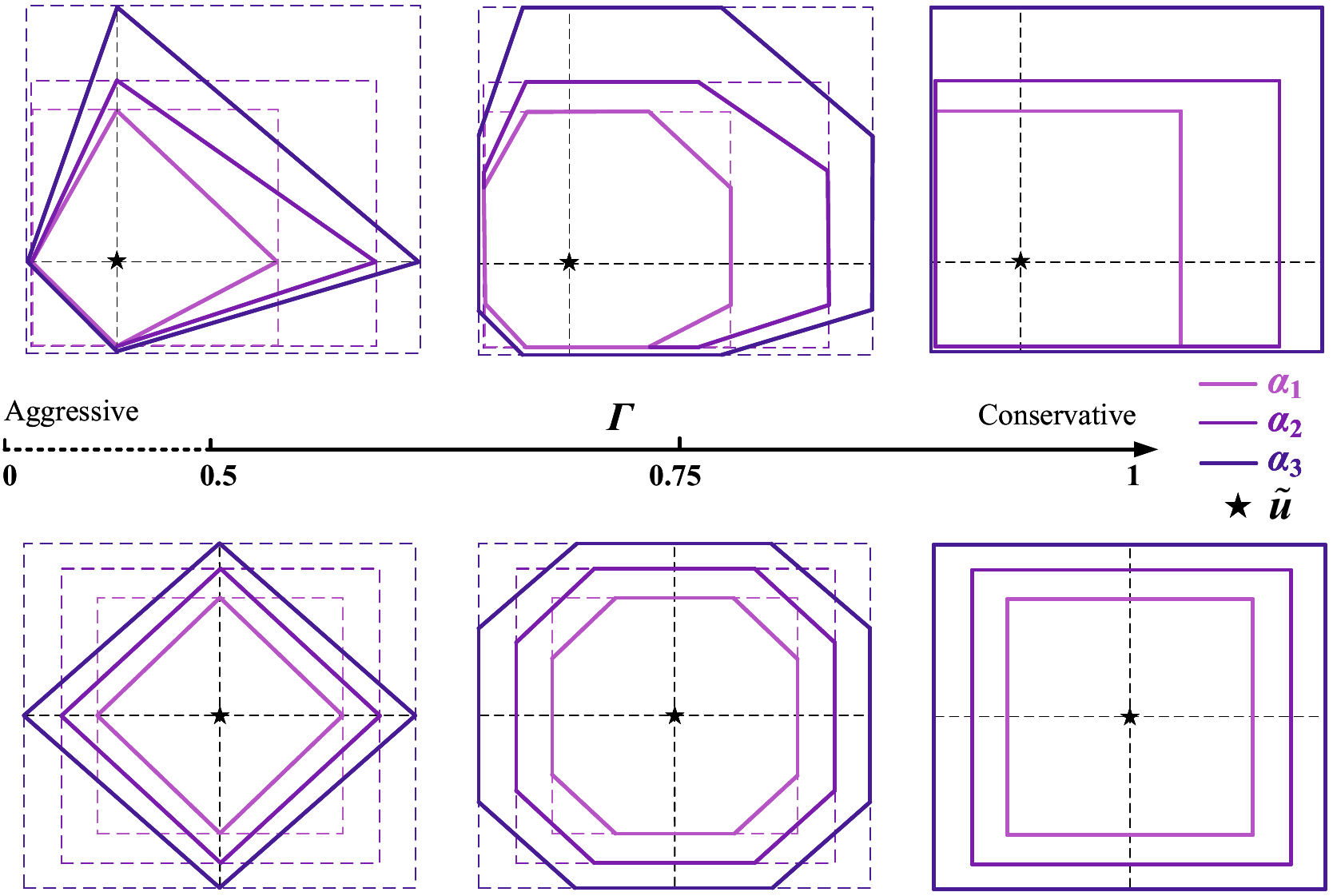}
\caption{Illustration of two-dimensional generalized uncertainty sets. The first row shows that $\mathcal{U}(\alpha)$ can represent asymmetric uncertainties, with off-centered $\tilde{\bm{u}}$ and nonlinear growth in $\alpha$, while the second row shows symmetric sets with potentially linear expansion. These patterns reflect underlying characteristics of uncertainty captured from historical data. With size controlled by $\alpha$ and shape by $\Gamma$, the proposed formulation represents diverse uncertainty geometries and structured dependencies. The first and third cases in the second row correspond to uncertainty modeling approaches in IGDT-based energy system studies in~\cite{w_igdt, idgt_binary1, idgt_binary2}, which are covered as special cases of $\mathcal{U}(\alpha)$.}
\label{fig:generalized Uset}
\end{figure}
\vspace{-10pt}
\subsection{Some Theoretical Properties} \label{subsec: properties}
CL-IGDT (functional) is structured as a two-stage optimization problem. The first-stage decision $\alpha$ yields a maximum $\mathcal{U}(\alpha)$ such that the remaining decision variables are feasible for all realizations in this set. The second-stage decision $\bm{y}(\cdot)$ is a function of~$\bm{u}$ that adapts to the uncertainty realization.

\begin{assumption} \label{assumption:recourse}
The second-stage problem admits relatively complete recourse, namely, the set 
\begin{align}
\mathcal{Y}(\bm{x},\bm{u}):=\{\bm{y}\in\mathbb{R}_+^{n_{\rm y}} \; |\; \bm{B}_2 \bm{y} \geq \bm{d} - \bm{B}_1 \bm{x} - \bm{E} \bm{u}\}  
\end{align}
is nonempty for each  
\begin{align}
\bm{x} \in 
\mathcal{X}:=\{\bm{x}\in\left\lbrace0,1 \right\rbrace^{n_{\rm x_1}}\times \mathbb{R}_+^{n_{\rm x_2}}\;|\;\bm{Ax}\leq\bm{b}\} \label{eq:cldigdt-nested-con2}
\end{align}
and $\bm{u} \in\mathcal U:=\cup_{\alpha\in[0,1]}\mathcal{U}(\alpha)$, and each recourse value
\begin{align}
W(\bm{x}, \bm{u})
:= \min_{\bm{y} \in \mathcal{Y}(\bm{x},\bm{u})} \bm{c}_2^\top \bm{y}
\label{eq:recourse_definition}
\end{align}
is always finite.
\end{assumption}

Thanks to nonnegativity of the variables $\hat{p}_{i,t}^{ls}$ and $\hat{p}_{i,t}^{cur}$, Assumption 1 is naturally satisfied for two-stage DN operation problems. While the functional form of CL-IGDT is conceptually clear, it is computationally intractable. We therefore convert it into an equivalent nested form, as is common in adaptive robust optimization reformulation~\cite{RO}.

\begin{prop}\label{prop-1}
\textbf{CL-IGDT (functional)} is equivalent to its nested form
\hspace{-0.8em} 

\noindent \textbf{CL-IGDT (nested):}
\begin{align}
\alpha^*(&\Lambda) = \max_{\substack{\alpha \in [0,1],\bm{x}\in\mathcal{X}}} \alpha \label{eq:cldigdt-nested-obj} \\
\textup{s.t.}~ 
& \bm{c}_1^\top \bm{x} + 
\max_{\bm{u} \in \mathcal{U}(\alpha)} 
\min_{\bm{y} \in \mathcal{Y}(\bm{x},\bm{u})} \bm{c}_2^\top \bm{y} \leq \Lambda.  \label{eq:cldigdt-nested-con1}
\end{align}
\end{prop}

To obtain an efficient solution strategy for CL-IGDT and support the algorithmic design in Section~\ref{sec: solution_methodology}, we first introduce a family of TSRO problems parameterized by $\alpha$ following~\cite{RO}. Specifically, for every $\alpha$,
\begin{align}
\textbf{TSRO:} \quad
\Lambda^*(\alpha) = \min_{\bm{x} \in \mathcal{X}}~ & \bm{c}^\top_{1} \bm{x} + \max_{\bm{u} \in \mathcal{U}(\alpha)}\min_{\bm{y} \in \mathcal{Y}(\bm{x},\bm{u})} \bm{c}^\top_{2} \bm{y}. \label{eq:tsro_nested-obj}
\end{align}

The value of the second-stage problem is denoted by
\begin{align}
    Q(\alpha, \bm{x}) =\max_{\bm{u} \in \mathcal{U}(\alpha)}\min_{\bm{y} \in \mathcal{Y}(\bm{x},\bm{u})} \bm{c}^\top_{2} \bm{y} \label{eq:Q_x}.
\end{align}

This bilevel program can be converted into
\begin{align}
     Q(\alpha, \bm{x}) = \max_{\bm{u} \in \mathcal{U}(\alpha), \bm{\pi} \in \Pi} \bm{(-Eu)}^\top \bm{\pi} + \bm{(d-B_1x)}^\top \bm{\pi}\label{eq:Q_x_dual}
\end{align}
by dualizing the inner LP. Here the feasible set of the dual variable $\bm\pi$ is given by
\begin{eqnarray}
&&\Pi= \{\bm{\pi}\in \mathbb{R}_+^{m_{\rm y}}\;|\;\bm{B}_2^\top\bm{\pi}\leq\bm{c}_2\}. \label{eq:dual_set}
\end{eqnarray}

We next establish the link between CL-IGDT and the family of TSRO problems in~\eqref{eq:tsro_nested-obj}. 

\begin{prop} \label{prop-2}
Under Assumption~\ref{assumption:recourse}, the optimal value $\alpha^*$ of \textbf{CL-IGDT (nested)} equals that of
\begin{align}
\textbf{CL-IGDT (single):} \quad
\alpha^* = \max_{\substack{\alpha \in [0,1]}} & \alpha \\
\textup{s.t.}\; 
& \Lambda^*(\alpha)\le\Lambda, 
\end{align}
where each $\Lambda^*(\alpha)$, $\alpha\in[0,1]$ is the optimal value of the corresponding TSRO problem in \eqref{eq:tsro_nested-obj}. Any optimal $\bm{x}^*$ of \eqref{eq:tsro_nested-obj} with $\alpha\equiv\alpha^*$ is also optimal for the CL-IGDT formulation.  
\end{prop}

\begin{proof}
We first show that there exists a feasible first-stage decision $\tilde{\bm x}\in\mathcal X$ for which the optimal value of the TSRO problem is attained, i.e.,
\begin{align}
\Lambda^*(\alpha)=
\bm{c}^\top_{1} \tilde{\bm x} + Q(\alpha, \tilde{\bm x}),
\end{align}
where $Q(\alpha,\bm x)$ is given by \eqref{eq:Q_x_dual} for each $\bm x\in\mathcal X$. 

By construction, $\mathcal U(\alpha)$ is compact. Under Assumption~\ref{assumption:recourse}, the optimal solution of $Q(\alpha,\bm x)$ is attained at extreme-point pairs of $\mathcal{U}(\alpha)$ and $\Pi$, so that
\begin{equation}
\begin{aligned}
&Q(\alpha,\bm x) = \\
&\max_{\substack{
\nu = 1,\dots,|\mathcal{P}_{\mathcal{U}(\alpha)}|,\\
\ell = 1,\dots,|\mathcal{P}_{\Pi}|
}}
(\bm{-E}\bm{u}^{\nu})^\top \bm{\pi}^{\ell}
+ \bm{(d-B_1\bm x)}^\top \bm{\pi}^{\ell}
\end{aligned}
\label{eq:Q_x_piecewise}
\end{equation}
where $\mathcal{P}_{\mathcal{U}(\alpha)}$ and $\mathcal{P}_{\Pi}$ denote the sets of extreme points of $\mathcal{U}(\alpha)$ and $\Pi$, respectively.

Hence, for any given $\alpha$, $Q(\alpha,\bm x)$ is piecewise affine in $\bm x$, and thus the objective function $c_1^\top \bm x + Q(\alpha,\bm x)$ is continuous on the compact set $\mathcal X$. Therefore, its optimal value is attained at some first-stage decision $\tilde{\bm x}\in\mathcal X$.

Next, assume that the pair $(\tilde\alpha,\tilde{\bm x})$ is feasible for CL-IGDT~(nested), i.e., 
\begin{align}
\bm{c}^\top_{1} \tilde{\bm{x}} + \max_{\bm{u} \in \mathcal{U}(\tilde\alpha)}\min_{\bm{y} \in \mathcal{Y}(\tilde{\bm{x}},\bm{u})} \bm{c}^\top_{2} \bm{y}\le\Lambda. 
\end{align}
Then $\Lambda^*(\tilde\alpha)\le \Lambda$, so $\tilde\alpha$ is also feasible for CL-IGDT~(single). Hence, the optimal value of CL-IGDT~(nested) does not exceed that of CL-IGDT~(single). 

Conversely, assume that $\tilde\alpha$ is feasible for CL-IGDT~(single). Since $\Lambda^*(\tilde\alpha)$ is attained at some $\tilde{\bm x} \in \mathcal{X}$, we have 
\begin{align}
\Lambda^*(\alpha)=\bm{c}^\top_{1} \tilde{\bm x}+ Q(\tilde\alpha,\tilde{\bm x})\le\Lambda.
\end{align}
Therefore, $(\alpha,\tilde{\bm x})$ is also feasible for CL-IGDT~(nested), and so the optimal value of CL-IGDT~(single) does not exceed that of CL-IGDT~(nested). This concludes the proof.  
\end{proof}

\section{Solution Methodology} \label{sec: solution_methodology}
The theoretical connection between CL-IGDT and TSROs enables the development of an efficient solution algorithm. According to Proposition~\ref{prop-2}, solving CL-IGDT reduces to identifying the largest~$\alpha$ for which the optimal value~$\Lambda^*(\alpha)$ of the corresponding TSRO satisfies the predefined budget constraint. This insight forms the basis of our solution strategy: searching for the maximum~$\alpha$ such that~$\Lambda^*(\alpha)$ meets the budget requirement. For several practically relevant distribution classes, such as unimodal ones, the confidence-level-based uncertainty set $\mathcal U(\alpha)$ can be chosen nested, i.e., $\mathcal U(\alpha_1)\subseteq\mathcal U(\alpha_2)$ for $\alpha_1\le\alpha_2$. Under this assumption, it follows from \eqref{eq:tsro_nested-obj} that $\Lambda^*(\alpha)$ is monotonic in $\alpha$, thereby facilitating a one-dimensional search. In this way, solving the CL-IGDT problem is transformed into solving a sequence of TSROs, avoiding direct reformulation of the bilevel structure, whose lower level embeds a sub-optimization problem associated with the generalized uncertainty set. An overview of the algorithmic workflow is illustrated in Fig.~\ref{fig:flowchart}.
\begin{figure}[!t]
\centering
\includegraphics[width=2.8in]{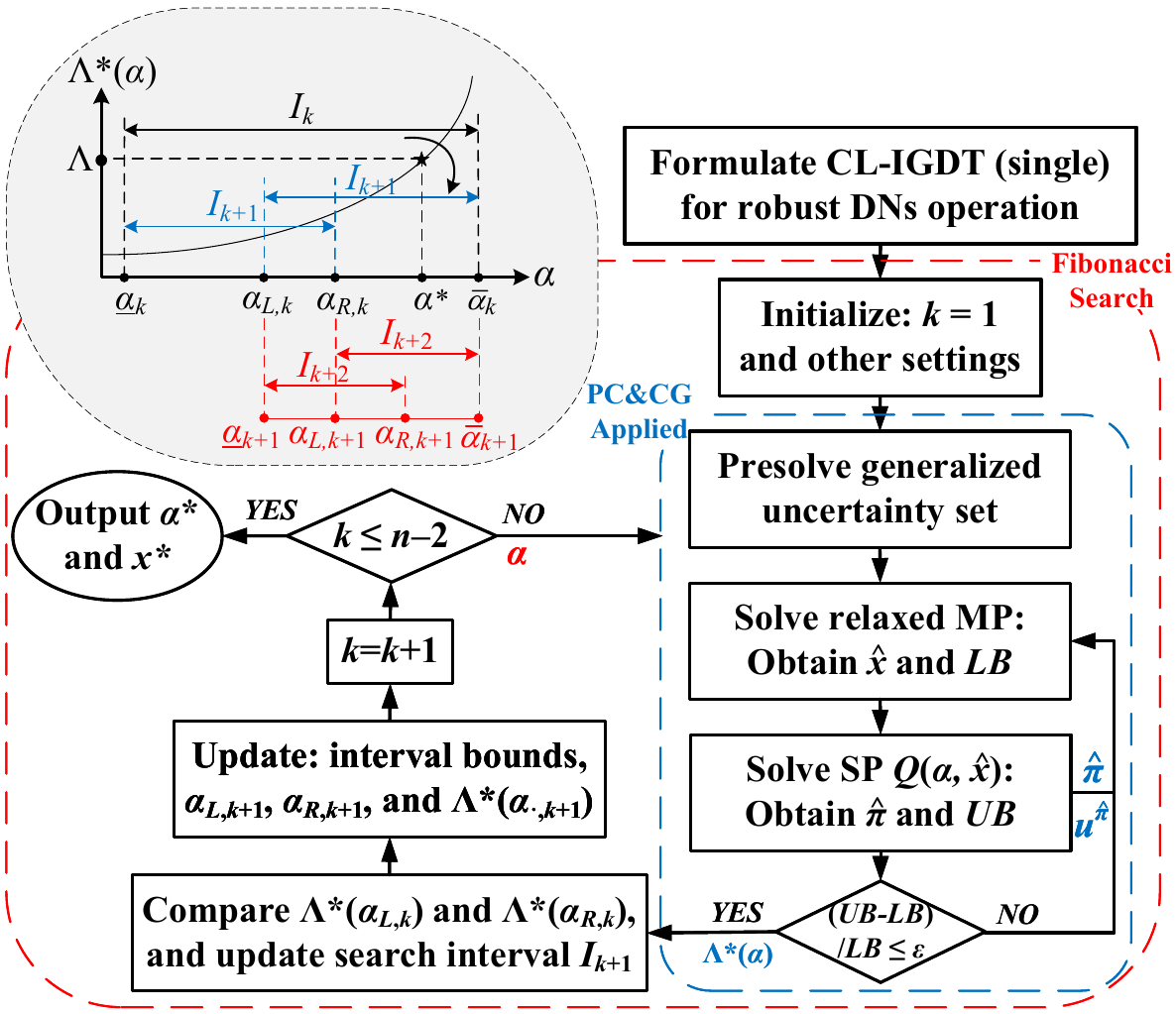}
\caption{Flowchart of F-PC\&CG algorithm.}
\label{fig:flowchart}
\end{figure}
\vspace{-10pt}
\subsection{Customized Fibonacci Search} \label{subsec:customized_fibanacci_search}

To locate the optimal $\alpha^*$, we propose a customized Fibonacci search tailored to CL-IGDT. The method is adapted from the classical Fibonacci search, which locates the minimizer of a unimodal function in unconstrained optimization and achieves efficient interval reduction through Fibonacci ratios~\cite{fibonacci_search}.
Here, however, the goal is to identify the maximum~$\alpha$ such that $\Lambda^*(\alpha) \leq \Lambda$. Accordingly, we retain the interval-shrinking mechanism while revising the interval selection rule.

As illustrated in Fig.~\ref{fig:flowchart}, the interval $I_k$ at iteration $k$ shrinks to $I_{k+1}$ after comparing $\Lambda^*(\alpha_{L,k})$ and $\Lambda^*(\alpha_{R,k})$ at two interior points determined by the Fibonacci sequence~$F_k$. The subinterval that does not contain $\alpha^*$ is discarded until the final interval~$I_n$ reaches a prescribed tolerance. The length of~$I_n$ after $n - 2$ iterations can be predetermined using~\eqref{eq:fib_def}--\eqref{eq:I_n}. For example, if the required tolerance is 0.5\%, i.e., $I_{n} \leq 0.005$, then with $I_{1} = 1$, we obtain $n = 13$ and $F_{13} = 233$, yielding $I_{13} \approx 0.0043$.
\begin{align}
&F_k =
\begin{cases}
1, & k = 0, 1, \\
F_{k-1} + F_{k-2}, & k \geq 2,
\end{cases} \label{eq:fib_def} \\
&I_{k+2} = \frac{F_{n-k-1}}{F_{n-k}} I_{k+1}, \label{eq:I_rec} \\
&I_n = \frac{I_1}{F_n}. \label{eq:I_n}
\end{align}

Since the decision variable also serves as the objective in CL-IGDT, prescribing a tolerance on $\alpha$ directly bounds the optimal value, making Fibonacci search particularly suitable. The customized interval selection rule is illustrated in Appendix~\ref{appen-2}. Specifically, the right interval $I_{k+1} = [\alpha_{L,k}, \overline{\alpha}_k]$ is retained if $\Lambda^*(\alpha_{L,k}) \le \Lambda^*(\alpha_{R,k}) < \Lambda$; otherwise, the left interval $I_{k+1} = [\underline{\alpha}_k, \alpha_{R,k}]$ is chosen.

This defines the $\alpha$-search mechanism for selecting the next candidate~$\alpha$ to evaluate. The value~$\Lambda^*(\alpha)$ for each candidate is then computed through the tailored TSRO procedure described in Sections~\ref{sec: dual-reformulation} and~\ref{sec: F-PCCG}.
\vspace{-12pt}
\subsection{Dual-based Reformulation for TSRO} \label{sec: dual-reformulation}
Each TSRO instance corresponds to an uncertainty boundary determined by~$\alpha$. This boundary defines~$\mathcal{U}_{CL}(\alpha)$ and is obtained through a discrete search based on the shortest confidence-interval construction in~\cite{zhishengxiong}, implemented here as a preprocessing step. Once the boundary is established, $\mathcal{U}(\alpha)$ becomes a polyhedral uncertainty set, so the resulting TSRO problem can be solved using the C\&CG algorithm~\cite{CCG}.

While classical C\&CG framework is applicable, its standard formulation does not explicitly expose a structure that can be exploited across different TSRO instances. To address this, we propose an alternative dual-based reformulation, whose advantages will be further elaborated in the next subsection.

Recall that $Q(\alpha, \bm{x})$ can be expressed as the bilinear optimization problem in~\eqref{eq:Q_x_dual}, where $\Pi$ is a fixed polyhedron independent of both $\alpha$ and $\bm{x}$. Given any $\bm\pi \in \mathcal{P}_\Pi$, the set of optimal $\bm{u}$ in~\eqref{eq:Q_x_dual} can be determined by the KKT optimality conditions of the LP $\max_{\bm{u} \in \mathcal{U}(\alpha)} \bm{(-Eu)}^\top \bm{\pi}$, namely,
\begin{align}
\mathcal{OU}(\alpha,\bm{\pi})
=\big\{
&\bm{u}\in\mathbb{R}_{+}^{n_{\rm u}},\ \bm{\lambda}\in\mathbb{R}_{+}^{m_{\rm u}}
\ \big|\ 
\bm{F}(\alpha)\bm{u}\le \bm{0}, \\
&\bm{E}^\top\bm{\pi}+\bm{F}^\top\bm{\lambda} \ge \bm{0}, \\
&\bm{\lambda}\circ\big(\bm{F}(\alpha)\bm{u}\big)=\bm{0}, \label{eq:ou-1} \\
&\bm{u}\circ\big(\bm{E}^\top\bm{\pi}+\bm{F}^\top\bm{\lambda}\big)=\bm{0}
\big\}, \label{eq:ou-2}
\end{align}
where $\bm\lambda$ is the Lagrangian multiplier of \eqref{eq:Uset}. 

The optimal $\bm{u}$ can be expressed parametrically in terms of $\alpha$ and $\bm\pi$. This insight is key to the reformulation in \textbf{Theorem~\ref{theo:1}}. 

\begin{theorem} \label{theo:1}
For each $\alpha\in[0,1]$, \textbf{TSRO} is equivalent to the  single-level optimization problem

\hspace{-0.8em} 
\noindent\textbf{TSRO-(single):}
\begin{align}
\Lambda^*(\alpha) &= \min \quad  \bm{c}^\top_1 \bm{x} + \eta \\
{\rm s.t.} \quad & \bm{A}\bm{x} \leq \bm{b}, \label{eq:x_set1}  \\
& \forall \bm{\pi} \in \mathcal{P}_{\Pi}: 
\begin{cases} 
\eta - \bm{c}_2^\top \bm{y}^{\bm\pi}\geq 0, \bm{y}^{\bm\pi} \in \mathbb{R}_+^{n_{\rm y}},  \\
\bm{B}_1 \bm{x} + \bm{B}_2 \bm{y}^{\bm\pi} \geq \bm{d} - \bm{E} \bm{u}^{\bm\pi}, \\
\bm{u}^{\bm{\pi}} \in \mathcal{OU}^{\bm{u}}(\alpha, \bm{\pi}),
\end{cases} \label{eq:cut} \\
& \bm{x} \in \{0, 1\}^{n_{\rm x_1}} \times \mathbb{R}_+^{n_{\rm x_2}}, \label{eq:x_set2}
\end{align}
where $\bm\pi$ indexes the significant scenarios, and $\mathcal{OU}^{\bm{u}}(\alpha, \bm{\pi})$ denotes the projection of the elements in $\mathcal{OU}(\alpha, \bm{\pi})$ onto their $\bm{u}$-space. 
\end{theorem}

TSRO-(single) is more accessible but extremely large scale, since full enumeration of all extreme points in~$\mathcal{P}_\Pi$ to construct the complete constraint sets~\eqref{eq:cut} is \textit{strongly NP-hard}. Fortunately, this exhaustive enumeration can be replaced by a partial one through a decomposition framework.

\begin{corollary} \label{cor:masterproblem}
Let $\hat{\mathcal{P}}_\Pi \subseteq \mathcal{P}_\Pi$. The master problem
\begin{align}
\hspace{-0.6em}\textbf{\textup{MP:}}\,\,\,\min~& \bm{c}^\top_1 \bm{x} + \eta \\
{\rm s.t.}~ & \eqref{eq:x_set1}, \eqref{eq:x_set2},  \\
& \forall \bm{\pi} \in \mathcal{\hat{P}}_{\Pi}: 
\begin{cases} 
\eta - \bm{c}_2^\top \bm{y}^{\bm\pi}\geq 0, \bm{y}^{\bm\pi} \in \mathbb{R}_+^{n_{\rm y}},  \\
\bm{B}_1 \bm{x} + \bm{B}_2 \bm{y}^{\bm\pi} \geq \bm{d} - \bm{E} \bm{u}^{\bm\pi}, \\
\bm{u}^{\bm{\pi}} \in \mathcal{OU}^{\bm{u}}(\alpha, \bm{\pi})
\end{cases} \label{eq:pccg_cut1}
\end{align}
is a relaxation of \textbf{TSRO-(single)}, with its optimal value being smaller than or equal to $\Lambda^*(\alpha)$. 
\end{corollary}

Therefore, the optimal solution $(\hat{\bm{x}}, \hat{\eta})$ of MP yields the lower bound (LB)
\begin{align}
LB = 
\bm{c}^\top_1 \hat{\bm{x}} + \hat{\eta}
\end{align}
of TSRO-(single). To strengthen the LB, $\hat{\mathcal{P}}_\Pi$ is iteratively augmented using information obtained from solving~$Q(\alpha, \hat{\bm{x}})$. Applying the KKT conditions to its inner LP converts the bilevel formulation of~$Q(\alpha, \hat{\bm{x}})$ into the single-level subproblem
\begin{align}
\hspace{-0.6em}\textbf{SP:}\,\,\,Q(\alpha, \hat{\bm{x}})=&\max~\bm{c}_2^\top \bm{y} \\
\text{s.t.}~ 
& \bm{F}(\alpha) \bm{u} \leq \bm{0}, \\
& \bm{B}_2 \bm{y} + \bm{E} \bm{u} \geq \bm{d} - \bm{B}_1 \hat{\bm{x}}, \\
& \bm{C} - \bm{B}_2^\top \bm{\pi} \geq \bm{0}, \\
& \bm{\pi} \circ \big(\bm{B}_2 \bm{y} + \bm{E} \bm{u} - \bm{d} + \bm{B}_1 \hat{\bm{x}} \big) = \bm{0}, \label{eq:kkt1} \\
&\bm{y} \circ \big( \bm{C} - \bm{B}_2^\top \bm{\pi} \big) = \bm{0},  \label{eq:kkt2} \\
& \bm{u} \in \mathbb{R}_+^{n_{\rm u}},\, \bm{y} \in \mathbb{R}_+^{n_{\rm y}},\, \bm{\pi} \in \mathbb{R}_+^{m_{\rm y}}.
\end{align}

By solving SP, the optimal $\hat{\bm\pi}$ and $\bm{u}^{\hat{\bm{\pi}}} \in \mathcal{OU}^{\bm{u}}(\alpha, \hat{\bm{\pi}})$ are obtained, and $\mathcal{\hat{P}}_{\Pi}$ is updated as $\mathcal{\hat{P}}_{\Pi} \cup \{\hat{\bm\pi}\}$. Subsequently, the new variable $\bm{y}^{\hat{\bm\pi}}$ and constraints \eqref{eq:pccg_cut1} with the obtained $\bm{u}^{\hat{\bm{\pi}}}$, all indexed by $\hat{\bm{\pi}}$, are added to MP as \textit{optimality cut sets}, strengthening the LB. Since $\hat{\bm{x}}$ is a feasible solution to TSRO-(single), the upper bound (UB) can be updated as
\begin{align}
    UB = \text{min}\{UB, \,\,\bm{c}^\top_1 \hat{\bm{x}} +  Q(\alpha, \hat{\bm{x}})\}.
\end{align}

\begin{remark}
The complementarity slackness in \eqref{eq:kkt1}--\eqref{eq:kkt2} can be linearized using a big-$M$ technique~\cite{CCG}, resulting in MIP formulations for both \textbf{MP} and \textbf{SP}, which can be solved directly using commercial solvers.
\end{remark}

The above iterative procedure progressively strengthens the lower and upper bounds through successive iterations, and terminates when the relative gap satisfies $\frac{UB - LB}{LB} \leq \epsilon$.
\vspace{-10pt}
\subsection{F-PC\&CG Algorithm} \label{sec: F-PCCG}
To solve CL-IGDT, the above procedure must be executed repeatedly throughout the $\alpha$-search rounds. In this brute-force implementation, previously generated cutting planes would be discarded in each new round, resulting in many redundant iterations. To address this inefficiency, we observe that although each TSRO instance admits the DIU property, from the perspective of the overall CL-IGDT problem, the uncertainty set exhibits a DDU structure. Recall that in Section~\ref{sec: dual-reformulation}, the optimality cuts are constructed from $\hat{\bm{\pi}}$ within a fixed polyhedron $\Pi$, which remains unchanged across TSRO instances.

Based on this observation, we propose a \textit{cut-recycling strategy} that exploits invariant dual information across TSROs, allowing previously generated optimality cuts to be recycled in subsequent rounds. Specifically, the cuts generated in round $(k-1)$ are reused to initialize round $k$, while being adapted to the updated $\alpha_k$, as
\begin{align}
&\forall \bm{\pi} \in \hat{\mathcal{P}}^{\rm init}_{\Pi}(k):
\begin{cases} 
\eta-\bm{c}_2^\top \bm{y}^{\bm\pi}\geq 0, \bm{y}^{\bm\pi} \in \mathbb{R}_+^{n_{\rm y}},  \\
 \bm{B}_1 \bm{x} + \bm{B}_2 \bm{y}^{\bm\pi} \geq \bm{d} - \bm{E} \bm{u}^{\bm\pi}, \\
\bm{u}^{\bm{\pi}} \in \mathcal{OU}^{\bm{u}}(\alpha_{k}, \bm{\pi}),
\end{cases} \label{eq:recycled_cut1} \\
&\hat{\mathcal{P}}^{\rm init}_{\Pi}(k)=\hat{\mathcal{P}}^{\rm fin}_{\Pi}(k-1), \label{eq:recycled_cut2}
\end{align}
where $\hat{\mathcal{P}}^{\rm init}_{\Pi}(k)$ and $\hat{\mathcal P}_{\Pi}^{\rm fin}(k-1)$ denote the set of initial extreme points of $\Pi$ at round $k$ and those obtained at the end of round $(k-1)$, respectively. With much less effort than repeatedly solving SP, the uncertain parameters $\bm{u}^{\bm{\pi}}$ can either be precomputed or embedded as decision variables.

In this way, all previously generated cutting planes are recycled to warm-start each new $\alpha$-search round. As the number of rounds increases, $\hat{\mathcal{P}}_{\Pi}$ accumulates additional dual information, making cut recycling increasingly beneficial and significantly reducing the total number of iterations. Compared to classical C\&CG, the recycled optimality cuts in \eqref{eq:recycled_cut1}-\eqref{eq:recycled_cut2} involve $\bm{u}$ in parametric form with respect to $\alpha$ and $\bm{\pi}$, giving rise to a "parametric C\&CG" framework. This idea is inspired by the PC\&CG algorithm~\cite{PCCG}, originally developed for DDU-based problems. By embedding this PC\&CG procedure into the outer-loop customized Fibonacci search, we obtain the Fibonacci-Parametric Column-and-Constraint Generation algorithm.

An overview of F-PC\&CG is provided in Fig.\ref{fig:flowchart}. The detailed steps are outlined in \textbf{Algorithm~1} and \textbf{Algorithm~2}, with asymptotic convergence and iteration complexity established in \textbf{Theorem 2}, where $n$ denotes the number of evaluations of $\Lambda^*(\alpha)$, corresponding to $n-2$ iterations of the Fibonacci search.

\begin{algorithm}[htbp]
    \small
    \caption{Customized Fibonacci Search (Outer Loop)}
    \label{alg:combined}

    \begin{algorithmic}[1]
    \State \textbf{Initialization:} Set $\Lambda$, $n\geq2$, $\underline{\alpha}_1 \gets 0$, $\overline{\alpha}_1 \gets 1$, $k \gets 1$.
    \State Compute $I_1 = \overline{\alpha}_1 - \underline{\alpha}_1$, and $I_2$ using \eqref{eq:I_rec}.
    \State Set $\alpha_{L,1} = \overline{\alpha}_1 - I_2$, \quad $\alpha_{R,1} = \underline{\alpha}_1 + I_2$.

    \While{$k \leq n-2$}
        \State \parbox[t]{0.9\linewidth}{
          Solve \textbf{TSRO-(single)} using PC\&CG algorithm to evaluate the unknown 
          $\Lambda^*(\alpha_{L,k})$ and/or $\Lambda^*(\alpha_{R,k})$.
        }
        \State Compute $I_{k+2}$ using \eqref{eq:I_rec}.
        \If{$\Lambda^*(\alpha_{L,k}) \leq \Lambda^*(\alpha_{R,k}) < \Lambda$}
            \State $\underline{\alpha}_{k+1} \gets \alpha_{L,k}$, \,\,\,\,\,\, $\overline{\alpha}_{k+1} \gets \overline{\alpha}_{k}$;
            \State $\alpha_{L,k+1} \gets \alpha_{R,k}$, \, $\alpha_{R,k+1} \gets \underline{\alpha}_{k+1} + I_{k+2}$;
            \State $\Lambda^*(\alpha_{L,k+1}) \gets \Lambda^*(\alpha_{R,k})$.
        \Else
            \State $\underline{\alpha}_{k+1} \gets \underline{\alpha}_k$, \, $\overline{\alpha}_{k+1} \gets \alpha_{R,k}$;
            \State $\alpha_{L,k+1} \gets \overline{\alpha}_{k+1} - I_{k+2}$,  \, $\alpha_{R,k+1} \gets \alpha_{L,k}$;
            \State $\Lambda^*(\alpha_{R,k+1}) \gets \Lambda^*(\alpha_{L,k})$.
        \EndIf
        \State $k \gets k + 1$.
    \EndWhile
    \State \textbf{Output:} Return $\alpha^* \gets \alpha_{L,k}$, and its corresponding $\bm{x}^*$.
    \end{algorithmic}

    \hrule height 0.4pt \kern 1pt \hrule height 0.4pt 
    \vspace{2pt}
    {\normalsize \noindent \textbf{Algorithm 2} Parametric C\&CG (Inner Loop)} 
    \vspace{2pt}
    \hrule height 0.4pt 

    \begin{algorithmic}[1]
    \State \textbf{Initialization:} Given $\alpha$, set $\epsilon > 0$, $LB \gets -\infty$, $UB \gets +\infty$;
    Initialize the recycled optimality cuts using $\{\hat{\mathcal{P}}_\Pi\}$.

    \While{$\frac{UB - LB}{LB} > \epsilon$}
        \State Solve \textbf{MP}.
        \If{\textbf{MP} is infeasible}
            \State Terminate and report infeasibility.
        \Else
            \State Obtain optimal $(\hat{\bm{x}}, \hat{\eta})$, and update $LB \gets \bm{c}^\top \hat{\bm{x}} + \hat\eta$.
        \EndIf
        \State Solve \textbf{SP:} $Q(\alpha, \hat{\bm{x}})$;
        \State Obtain optimal $\hat{\bm{\pi}}$, and update $\hat{\mathcal{P}}_\Pi \gets \hat{\mathcal{P}}_\Pi \cup \{\hat{\bm{\pi}}\}$;
        \State Add corresponding constraints \eqref{eq:pccg_cut1} to \textbf{MP};
        \State Update $UB \gets \text{min}\{UB, \,\,\bm{c}^\top_1 \hat{\bm{x}} +  Q(\alpha, \hat{\bm{x}})\}$.
    \EndWhile
    \State \textbf{Output:} Return $\Lambda^*(\alpha) \gets \bm{c}^\top \hat{\bm{x}} + \hat\eta$ to Outer Loop.
    \end{algorithmic}
\end{algorithm}

\begin{theorem} \label{theo:2}
(1) The F-PC\&CG algorithm asymptotically converges to the optimal value and solution of CL-IGDT. \\
(2) The approximation error of the objective value is bounded by $1/F_n$, the number of iterations is bounded by $|\mathcal{P}_{\Pi}|+n$, and the algorithm is of $\mathcal{O}\big(\binom{m_y+n_y}{n_y}\big)$ iteration complexity.  
\end{theorem}


\begin{remark}
F-PC\&CG features a structural decoupling: its total iteration count is bounded by $n + |\mathcal{P}_\Pi|$, rather than $n|\mathcal{P}_\Pi|$ without introducing the cut recycling strategy. Since $|\mathcal{P}_\Pi|$ dominates the iteration complexity, the outer loop no longer drives overall iteration growth, thus significantly improving computational efficiency.
\end{remark}

\section{Case Study and Discussion} \label{Case Study}
The case study is conducted on a modified IEEE-33 bus system, as illustrated in Fig.~\ref{fig:IEEE33}. The main system parameters are available in~\cite{33_bus}, while DG and ESS data are given in Table~\ref{tab:dg_data} and Table~\ref{tab:ESSs_data}. The line capacity is set to 3.5~kW, and the voltage magnitude is constrained within [0.95, 1.05]. The time step is one hour over a 24-hour horizon. Approximately one year of historical data is used to construct the generalized uncertainty sets with $\Gamma=0.8$. Since $\alpha \in [0,1]$ and considering interval resolution, the search interval $I_8=0.048$ is adopted, which is sufficiently accurate. The PC\&CG tolerance is set to 0.5\%. All simulations are performed on a 64-bit PC (3.00~GHz CPU, 16~GB RAM) using the Gurobi API in PyCharm.
\begin{figure}[!t]
\centering
\includegraphics[width=3.0in]{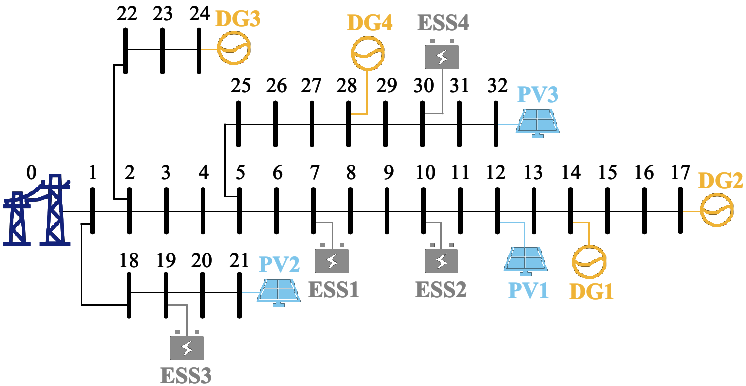}
\caption{Modified IEEE 33-bus system.}
\label{fig:IEEE33}
\end{figure}
\begin{table}[t]
	\centering
	\caption{DGs Information} 
	\vspace{-2mm}
	\scalebox{0.7}{
	\begin{tabular}{cccccccc}
		\toprule
		Label & $\overline{p}_{i}^{G}$(kW) & $\underline{p}_{i}^{G}$(kW) & $RU_{i}^G/RD_{i}^G$(kW) & $\overline{q}_{i}^{G}$(kW) & $\underline{q}_{i}^{G}$(kW) & $c_{i}^G$(\$) & $c_{i}^U$(\$)\\
		\midrule
		DG1 & 400 & 50 & 100 & 320 & -200 & 2.0 & 200 \\
		
		DG2 & 500 & 80 & 150 & 400 & -350 & 2.5 & 400 \\
		
		DG3 & 800 & 100 & 200 & 700 & -500 & 3.5 & 700 \\
  		
		DG4 & 1200 & 120 & 250 & 1100 & -800 & 4.0 & 1000 \\
		\bottomrule
	\end{tabular}}
	\label{tab:dg_data}
\vspace{-0.2em}    
\end{table}
\begin{table}[t]
	\centering
	\caption{ESSs Information} \vspace{-2mm}
	\scalebox{0.7}{
	\begin{tabular}{cccccc}
		\toprule
		Label & $\overline{p}_{i}^{B+}/\overline{p}_{i}^{B-}$(kW) & $\overline{e}_{i}$(kW) & $\underline{e}_{i}$(kW) & $\eta_{i}^{B}$  \\
		\midrule
		ESS1 & 120 & 700 & 140 & 0.9 \\
		
		ESS2 & 200 & 800 & 160 & 0.9 \\
		
		ESS3 & 250 & 900 & 180 & 0.9 \\
  		
		ESS4 & 280 & 1000 & 200 & 0.9 \\
		\bottomrule
	\end{tabular}}
	\label{tab:ESSs_data}
\end{table}
\vspace{-10pt}
\subsection{Optimal Operation Results}
The financial budget of the proposed CL-IGDT framework is specified as \$\num{259635.38}, i.e., 25\% above the optimal cost under the expected scenarios. The solution yields $\alpha^*=0.691$, with $\Lambda^*(\alpha^*)=\num{260472.86}$ deviating from the preset budget by 0.32\%. In comparison, the IGDT-based approach with the same budget yields an optimal uncertainty horizon $\delta^*=0.184$, representing the symmetric uncertainty modeling discussed in Section~\ref{sec:intro}. 

The comparison is shown in Fig.~\ref{fig:Uncertainty_Uset}. From the confidence bands of PV generation and load demand over 24 hours, $\mathcal{U}(\alpha^*)$ can be either symmetric or asymmetric around the predictive value. For PV generation, $\mathcal{U}(\alpha^*)$ is consistently larger than the IGDT-based set, whereas for load demand its size varies across time slots (e.g., 10h vs. 20h). This is because the probability-based set captures distributional features from historical data. As a result, although governed by a single variable, $\mathcal{U}(\alpha^*)$ remains source- and time-specific, enabling a finer characterization of uncertainty in DNs. By contrast, the IGDT-based set is always symmetric and centered on the predictive value, with its range controlled by a single~$\delta$. For instance, its range is narrow at 6h but wide at 20h, consistent with the predictive value level.

The embedded histograms at 11h provide further insight. For PV, nearly 45\% of historical samples concentrate at zero. With $\alpha^*=0.691$, $\mathcal{U}(\alpha^*)$ yields the shortest confidence interval covering around 69.1\% of the most probable scenarios, resulting in a non-trivial uncertainty set $[0,35]$. By contrast, the IGDT-based set degenerates to zero width because the predictive value is zero. Likewise, for load demand, $\mathcal{U}(\alpha^*)$ yields the interval $[92,178]$ covering the most probable scenarios, whereas the IGDT-based set remains centered on the predictive value and excludes some more likely realizations.

These results show that the CL-IGDT framework allocates the financial budget to cover the most significant scenarios while discarding trivial ones according to their probability of occurrence, thereby achieving an improved trade-off between economic performance and robustness in DN operations.
\begin{figure}[!t]
\centering
\subfigure[\label{fig:PV_Uset}]{\includegraphics[width=0.32\textwidth]{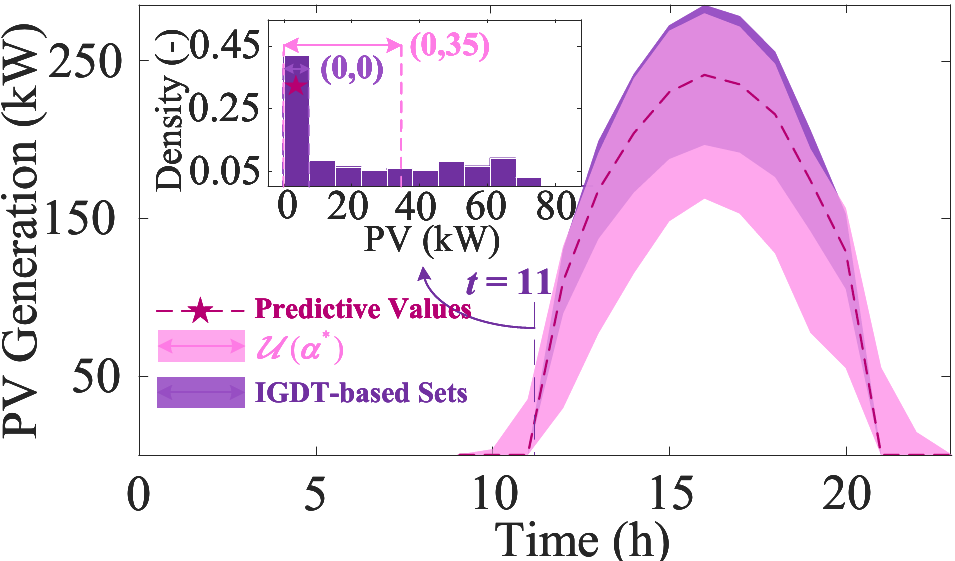}}
\subfigure[\label{fig:Load_Uset}]{\includegraphics[width=0.32\textwidth]{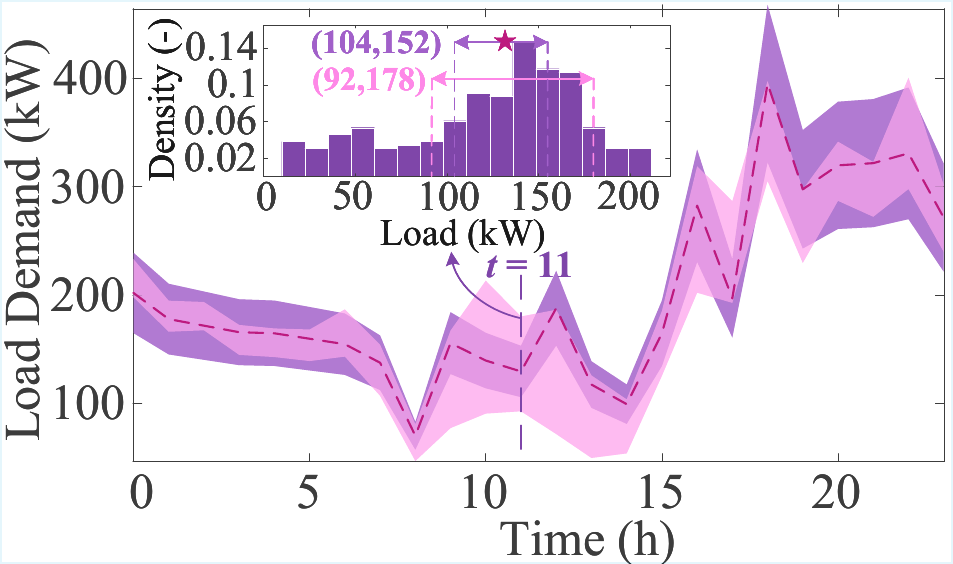}}
\caption{Optimal uncertainty sets: (a) PV generation at bus 12; (b) load demand at bus 7.}
\label{fig:Uncertainty_Uset}
\vspace{-0.3em}
\end{figure}
\vspace{-10pt}
\subsection{Out-of-Sample Performance}
\subsubsection{Comparison with Different Approaches} To evaluate the out-of-sample performance of the first-stage decisions~$\bm{x}^*$ derived from the proposed framework, two baseline methods are adopted for comparison: an IGDT-based robust optimization and a two-stage stochastic programming (TSSP) framework constructed via the sample average approximation (SAA)~\cite{saa}. To ensure a fair comparison, the random scenarios used in TSSP are sampled from the same uncertainty boundary defined by the optimized~$\mathcal{U}(\alpha^*)$. All three approaches are tested against the same set of 100 out-of-sample scenarios.
\begin{table}[t]
    \centering
    \caption{Out-of-sample Performance of Different Approaches}
    \vspace{-2mm}
    \scalebox{0.75}{
    \begin{threeparttable}
    \begin{tabular}{ccccc}
        \toprule
        Approach & $C_{\text{I}}(\$)$ & $\mathbb{E}[C_{\text{II}}] \pm \mathrm{Std}(\$)$ & $C_{\text{TOL}}(\$)$ & $\text{NoP}$ \\
        \midrule
        CL-IGDT  & 223{,}116.97 & 10{,}734.62\,±\,9{,}102.99 & 233{,}851.59 & 0 \\
        IGDT      & 218{,}217.04 & 26{,}564.42\,±\,17{,}766.33 & 244{,}781.46 & 12 \\
        TSSP      & 213{,}379.55 & 20{,}689.66\,±\,9{,}537.71 & 234{,}069.22 & 6 \\
        \bottomrule
    \end{tabular}
    \begin{tablenotes}
        \footnotesize
        \item Note: NoP denotes the number of penalty actions.
    \end{tablenotes}
    \end{threeparttable}
    }
    \label{tab:out_of_sample}
    \vspace{4pt}
    \vspace{-1em}
\end{table}

As reported in Table~\ref{tab:out_of_sample}, CL-IGDT yields the highest first-stage cost~$C_{\text{I}}$ (\$\num{223116.97}) compared to IGDT (\$\num{218217.04}) and TSSP (\$\num{213379.55}). However, this upfront investment pays off in the out-of-sample test: the average second-stage cost~$\mathbb{E}[C_{\text{II}}]$ under CL-IGDT is only \$\num{10734.62}, representing a reduction of \num{59.59}\% and \num{48.12}\% compared to IGDT and TSSP, respectively. Moreover, CL-IGDT achieves the lowest standard deviation ($\mathrm{Std}$) of $C_{\text{II}}$ (\$\num{9102.99}) and completely avoids penalty actions, whereas IGDT and TSSP result in 12 and 6 occurrences, respectively. This shows that CL-IGDT better protects the DNs against highly probable uncertain realizations by leveraging the distributional characteristics embedded in the generalized uncertainty set. In contrast, the uncertainty modeling in IGDT is not data-informed, so although its first-stage cost is lower, it provides insufficient protection against high-probability deviations from expectation.

TSSP also performs well because its scenario sampling is conducted over the optimized set~$\mathcal{U}(\alpha^*)$, indirectly benefiting from the modeling strength of CL-IGDT. However, unlike the proposed framework, which considers all scenarios within~$\mathcal{U}(\alpha^*)$, TSSP only accounts for a finite subset. As a result, it still misses some uncertain cases, leading to six penalty actions and a higher average second-stage cost. In summary, the proposed CL-IGDT framework shows clear advantages over traditional robust and stochastic optimization in both economic efficiency and operational reliability. It also provides guidance for uncertainty boundary selection in robust optimization.
\begin{table}[t]
    \centering
    \caption{Out-of-sample Performance under Varying $\Gamma$}
    \vspace{-2mm}
    \scalebox{0.75}{
    \begin{tabular}{cccc}
        \toprule
        $\Gamma$ & $C_{\text{I}}(\$)$ & $\mathbb{E}[C_{\text{II}}] \pm \mathrm{Std}(\$)$ & $C_{\text{TOL}}(\$)$ \\
        \midrule
        0.5 & 220{,}285.86 & 15{,}244.90\,±\,9{,}341.45 & 235{,}530.76 \\
        0.6 & 222{,}934.56 & 6{,}679.60\,±\,9{,}167.32 & 234{,}614.16 \\
        0.7 & 231{,}929.92 & 2{,}769.56\,±\,9{,}209.03 & 234{,}699.49 \\
        0.8 & 223{,}116.97 & 10{,}734.62\,±\,9{,}102.99 & 233{,}851.59 \\
        0.9 & 222{,}968.17 & 10{,}858.46\,±\,9{,}104.22 & 233{,}987.28 \\
        1.0 & 222{,}130.49 & 12{,}664.06\,±\,9{,}112.42 & 234{,}794.55 \\
        \bottomrule
    \end{tabular}}
    \label{tab:gamma_performance}
    \vspace{4pt}
\end{table}

\subsubsection{Sensitivity Analysis of~$\Gamma$} To demonstrate the enhanced expressiveness of the proposed generalized uncertainty set and how the resulting operational scheme $\bm{x}^*$ responds to different structural forms of $\mathcal{U}(\alpha^*)$, a sensitivity analysis on the structural parameter $\Gamma$ is conducted. As illustrated in Fig.~\ref{fig:generalized Uset}, the structure of $\mathcal{U}(\alpha)$ evolves to allow more temporally decoupled realizations and wider uncertainty coverage.

As shown in Table~\ref{tab:gamma_performance}, this structural shift directly affects the first-stage cost $C_{\text{I}}$, which first increases and then decreases. Specifically, $C_{\text{I}}$ increases from \$\num{220285.86} at $\Gamma = \num{0.5}$ to a peak of \$\num{231929.92} at $\Gamma = \num{0.7}$, before decreasing to \$\num{222130.49} at $\Gamma = \num{1.0}$. Initially, as temporal coupling weakens, $\bm{x}^*$ becomes more conservative due to increased redundant resource allocation across time slots, resulting in a higher $C_{\text{I}}$. However, when $\Gamma$ becomes too large, the model starts prioritizing the reduction of $C_{\text{I}}$ to control total operational cost. In contrast, the out-of-sample performance of $\bm{x}^*$, reflected by $\mathbb{E}[C_{\text{II}}]$, exhibits the opposite trend: it first decreases and then increases. As shown in Table~\ref{tab:gamma_performance}, $\mathbb{E}[C_{\text{II}}]$ declines from \$\num{15244.90} at $\Gamma = \num{0.5}$ to a minimum of \$\num{2769.56} at $\Gamma = \num{0.7}$, and then increases to \$\num{12664.06} at $\Gamma = \num{1.0}$. This is because the increasing range of $\mathcal{U}(\alpha)$ improves the robustness of $\bm{x}^*$ and its adaptability to out-of-sample scenarios. However, when $\Gamma$ becomes too large, the decoupled uncertainty modeling may deviate from the temporal correlation patterns implicit in the samples, leading to a misaligned $\bm{x}^*$.

As a result, the total cost $C_{\text{TOL}}$ first decreases and then increases as $\Gamma$ varies, reaching its minimum value of \num{233851.59} at $\Gamma = \num{0.8}$. However, there is no universally optimal $\Gamma$, as its choice reflects the decision-maker’s preference regarding temporal correlation. This result highlights how $\mathcal{U}(\alpha)$ enables a continuous structural transition from time-coupled worst-case realizations to fully decoupled ones. By tuning $\Gamma$, the decision-maker can shift between aggressive and conservative strategies, thereby adapting to diverse uncertainty characteristics in practice.
\vspace{-22pt}
\subsection{Computational Performance}
To demonstrate the superiority of the proposed F-PC\&CG, the baseline F-C\&CG algorithm is adopted for comparison. For F-C\&CG, the outer $\alpha$-search follows the same Fibonacci procedure, while the inner evaluation of $\Lambda^*(\alpha)$ is carried out by the C\&CG algorithm. Similar to PC\&CG, once a significant scenario $\hat{\bm{u}} \in \mathcal{P}_{\mathcal{U}(\alpha)}$ is obtained from solving SP, a new variable $\bm{y}^{\hat{\bm{u}}}$ is introduced and constraints~\eqref{eq:ccg_cut} are added to MP as optimality cuts, strengthening its LB. 
\begin{align}
&\forall \hat{\bm{u}} \in \hat{\mathcal{P}}_{\mathcal{U}(\alpha)}:
\begin{cases} 
\eta - \bm{c}_2^\top \bm{y}^{\hat{\bm{u}}}\geq 0, \bm{y}^{\hat{\bm{u}}} \in \mathbb{R}_+^{n_{\rm y}},  \\
\bm{B}_1 \bm{x} + \bm{B}_2 \bm{y}^{\hat{\bm{u}}} \geq \bm{d} - \bm{E} \hat{\bm{u}},
\end{cases} \label{eq:ccg_cut}
\end{align}
where the superscript $\hat{\bm{u}}$ indexes the significant scenarios, and $\hat{\mathcal{P}}_{\mathcal{U}(\alpha)}$ represents the subset of extreme points of $\mathcal{U}(\alpha)$. 

\begin{table}[t]
	\centering
	\caption{Computational Results of F-PC\&CG vs. F-C\&CG}
    \vspace{-2mm}
	\scalebox{0.75}{
	\begin{threeparttable}
	\begin{tabular}{ccc cc cc}
		\toprule
        \multirow{2}{*}{\raisebox{-0.7\height}{\makecell{Outer\\Iteration}}}
        & \multirow{2}{*}{\raisebox{-0.7\height}{$\alpha$}} 
        & \multirow{2}{*}{\raisebox{-0.7\height}{\makecell{Interval\\Selection}}} 
		& \multicolumn{2}{c}{PC\&CG} 
		& \multicolumn{2}{c}{C\&CG} \\
		\cmidrule(lr){4-5} \cmidrule(lr){6-7}
		 &  &  & \multicolumn{1}{c}{$\Lambda^*(\alpha)$} & \multicolumn{1}{c}{Gap} & \multicolumn{1}{c}{$\Lambda^*(\alpha)$} & \multicolumn{1}{c}{Gap} \\
		\cmidrule(r){1-3} \cmidrule(lr){4-5} \cmidrule(l){6-7}
		\multirow{2}{*}{1} & 0.619 & \multirow{2}{*}{$I_{R}$} & 248{,}741.12 & 0.05\% & 249{,}556.11 & 0.08\% \\
		  & 0.381 &  & 231{,}243.92 & 0.43\% & 231{,}680.59 & 0.38\% \\
		2 & 0.762 & $I_{L}$ & 270{,}904.41 & 0.49\% & 270{,}824.86 & 0.45\% \\
		Pruned & 0.524 & $I_{R}$ & -- & -- & -- & -- \\
		3 & 0.667 & $I_{R}$ & 257{,}411.52 & 0.33\% & 257{,}337.41 & 0.31\% \\
		4 & 0.714 & $I_{L}$ & 264{,}151.90 & 0.21\% & 263{,}532.21 & 0.37\% \\
		Pruned & 0.666 & $I_{R}$ & -- & -- & -- & -- \\
		5 & 0.691 & -- & 260{,}472.86 & 0.07\% & 260{,}639.67 & 0.14\% \\
		\bottomrule
	\end{tabular}
	\begin{tablenotes}
		\footnotesize
		\item Note: The expected scenarios have been employed to generate the initial optimality cuts for both algorithms, so as to strengthen the initial LB.
	\end{tablenotes}
	\end{threeparttable}
	}
	\label{tab:F-PCCG}
    \vspace{-0.5em}
\end{table}
Compared to~\eqref{eq:cut}, the scenario $\hat{\bm{u}}$ in the baseline algorithm is a fixed value rather than a parametric expression. As illustrated in Fig.~\ref{fig:generalized Uset}, the extreme points of $\mathcal{U}(\alpha)$ vary substantially due to its DDU property. Specifically, an extreme point derived from $\mathcal{U}(\alpha_1)$ is generally not an extreme point of $\mathcal{U}(\alpha_2)$ when $\alpha_1 \neq \alpha_2$, and may not even belong to $\mathcal{U}(\alpha_2)$. Consequently, the C\&CG procedure must be re-executed from scratch for each $\alpha$ within the same problem instance, resulting in considerable redundant computation. This directly explains the computational results reported in Table~\ref{tab:F-PCCG} and the convergence trends in Fig.~\ref{fig:convergence}.

\begin{figure}[!t]
\centering
\includegraphics[width=3.0in]{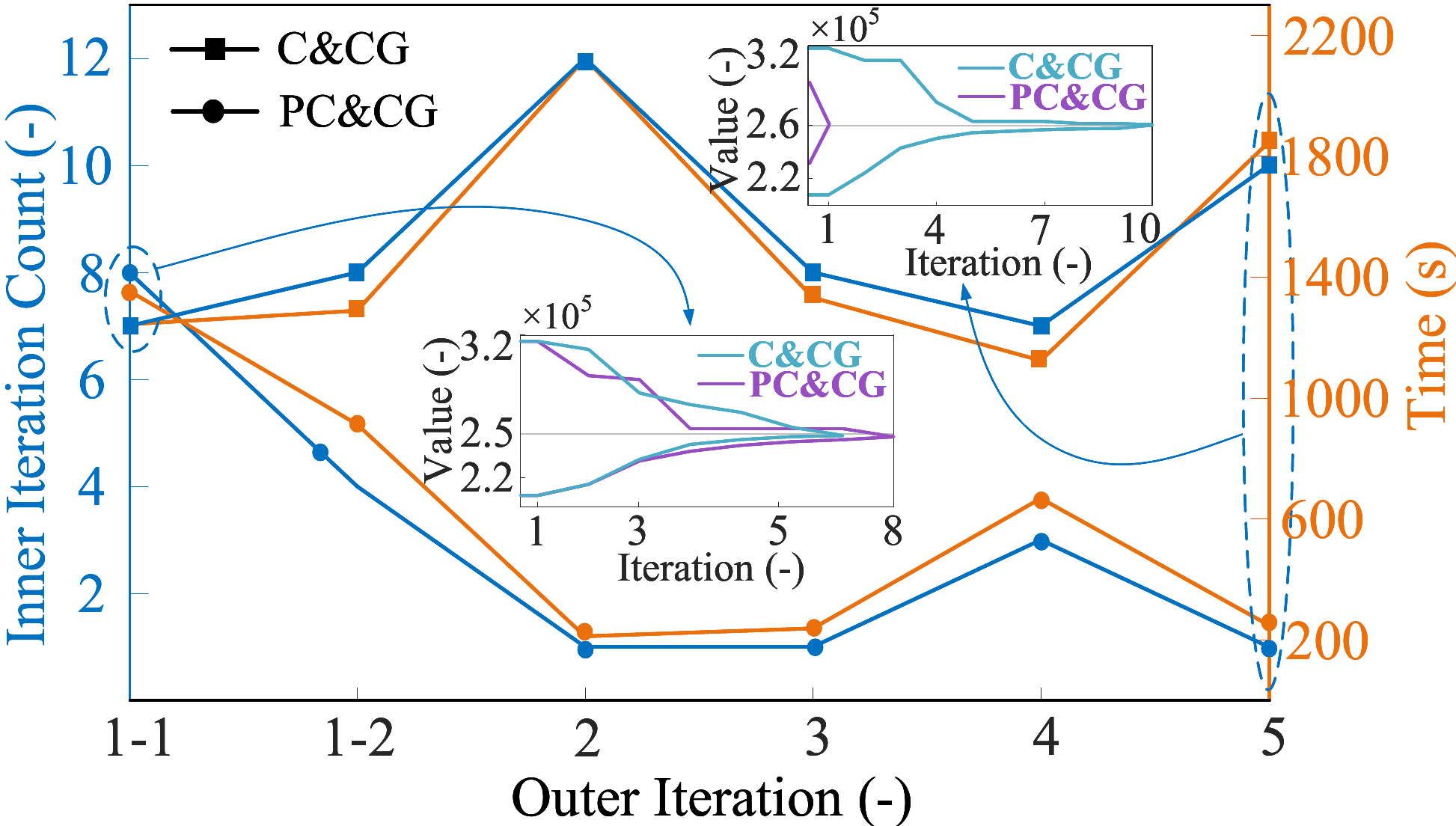}
\caption{Convergence behaviors of PC\&CG and C\&CG. At the first Fibonacci round, two function evaluations are required, which are denoted as 1-1 and 1-2.}
\label{fig:convergence}
\end{figure}

The computational results of solving the CL-IGDT problem are detailed in Table~\ref{tab:F-PCCG}. In total, six subinterval selections are made, generating eight candidate values of~$\alpha$, of which six are explicitly evaluated and two are pruned. Given the resolution $I_8=0.048$, the maximum deviation between the obtained $\alpha^*=0.691$ and the exact optimum is at most 0.048. Across all outer Fibonacci rounds, both algorithms converge to nearly the same value, with relative gaps consistently within the prescribed tolerance.

Fig.~\ref{fig:convergence} illustrates the convergence behaviors of PC\&CG and C\&CG across outer iterations. At the first function evaluation, both algorithms require a similar number of inner iterations and comparable runtime. From the second evaluation onward, however, the inner iterations required for PC\&CG drop sharply because all cuts generated in previous rounds are recycled, enabling a direct warm-start and yielding a much stronger initial LB. In particular, at the second, third, and fifth outer iterations, only a single inner iteration is needed. By contrast, in C\&CG each outer iteration is equivalent to solving a new problem, so the inner iteration counts remain consistently higher.

In terms of computational time, the average runtime per inner iteration is comparable for both algorithms. Nevertheless, the total iteration count for solving the CL-IGDT problem is $18$ for F-PC\&CG, compared to $52$ for F-C\&CG. This reduction translates into a substantial decrease in overall runtime (a total reduction of $5{,}326\,\mathrm{s}$), highlighting the clear computational advantage of the proposed algorithm. If more Fibonacci search rounds were required, the superiority of F-PC\&CG would become even more pronounced. In practical multi-day DN operation, the generated cuts can be further recycled to enhance computational efficiency.
\vspace{-20pt}
\section{Conclusions}
This paper proposed a generalized uncertainty set within the CL-IGDT framework for the robust operation of distribution networks. This formulation not only enables a more flexible characterization of RESs and load uncertainties, but also enhances the modeling capability of IGDT-based methods, thereby supporting more tailored operational strategies. By uncovering the theoretical equivalence between CL-IGDT and TSROs, and recognizing the DDU property of the proposed set, we introduced a cut-recycling strategy and developed a novel F-PC\&CG algorithm that guarantees asymptotic convergence with bounded approximation error. Case studies have validated the effectiveness and advantages of the proposed framework and algorithm.
\vspace{-10pt}
\appendices
\section{Illustration of the Interval Selection Rule} \label{appen-2}
The representative examples in Fig.~\ref{fig:interval_selection} demonstrate how the revised interval selection rule operates. Apart from this, all other procedures of Fibonacci search remain unchanged. Case~1 shows that $\alpha^*$ lies in the right subinterval $I_R = [\alpha_L,\, \overline{\alpha}]$ when $\Lambda^*(\alpha_L) \leq \Lambda^*(\alpha_R) < \Lambda$. In case~2, either subinterval may be used when $\Lambda^*(\alpha_L) < \Lambda < \Lambda^*(\alpha_R)$. In case~3, the left subinterval $I_L = [\underline{\alpha},\, \alpha_R]$ is selected when $ \Lambda < \Lambda^*(\alpha_L) \leq \Lambda^*(\alpha_R)$. As a special case, if a flat segment may occur, i.e., whenever $\Lambda^*(\alpha_L)=\Lambda$ and/or $\Lambda^*(\alpha_R)=\Lambda$, we retain $I_R$ to ensure the largest feasible $\alpha^*=\max\{\alpha:\Lambda^*(\alpha)\le\Lambda\}$ remains in the selected interval.

Furthermore, exploiting the monotonicity of $\Lambda^*(\alpha)$ yields a pruning criterion that allows one to skip a function evaluation and directly select the next subinterval. Specifically, in the $k$th Fibonacci search, given $\Lambda^*(\alpha_{R,k}) < \Lambda$, the unevaluated $\Lambda(\alpha_{L,k})$ satisfies $\Lambda^*(\alpha_{L,k}) \leq \Lambda^*(\alpha_{R,k}) < \Lambda$, and $I_{R,k+1}$ can be selected without additional evaluation. Similarly, given $\Lambda < \Lambda^*(\alpha_{L,k})$, $I_{L,k+1}$ can be selected without additional evaluation.

\begin{figure}[!t]
\centering
\includegraphics[width=3.0in]{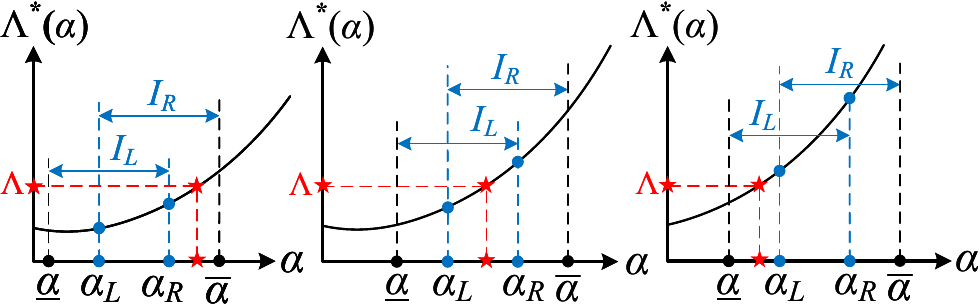}
\caption{Three representative examples of search interval reduction.}
\label{fig:interval_selection}
\end{figure}

\ifCLASSOPTIONcaptionsoff
  \newpage
\fi
\vspace{-6pt}
\bibliographystyle{IEEEtran}
\bibliography{reference}

@article{zhishengxiong,
      title={Two-Stage Robust Optimal Operation of Distribution Networks Considering Renewable Energy and Demand Asymmetric Uncertainties}, 
      author={Zhisheng Xiong and Bo Zeng and Peter Palensky and Pedro P. Vergara},
      journal={arXiv preprint arXiv:2411.10166},
      year={2025}
}

@article{PCCG,
  title={Two-stage robust optimization with decision dependent uncertainty},
  author={Zeng, Bo and Wang, Wei},
  journal={arXiv preprint arXiv:2203.16484},
  year={2022}
}

@ARTICLE{yangyue,
  author={Yang, Yue and Wu, Wenchuan},
  journal={IEEE Trans. Smart Grid}, 
  title={A Distributionally Robust Optimization Model for Real-Time Power Dispatch in Distribution Networks}, 
  year={2019},
  volume={10},
  number={4},
  pages={3743-3752},
  doi={10.1109/TSG.2018.2834564}}

@ARTICLE{active_power_ESS,
  author={Vergara, Pedro P. and López, Juan C. and Rider, Marcos J. and da Silva, Luiz C. P.},
  journal={IEEE Trans. Smart Grid}, 
  title={Optimal Operation of Unbalanced Three-Phase Islanded Droop-Based Microgrids}, 
  year={2019},
  volume={10},
  number={1},
  pages={928-940}}

@ARTICLE{lossy_DistFlow,
  author={Schweitzer, Eran and Saha, Shammya and Scaglione, Anna and Johnson, Nathan G. and Arnold, Daniel},
  journal={IEEE Trans. Power Syst.}, 
  title={Lossy DistFlow Formulation for Single and Multiphase Radial Feeders}, 
  year={2020},
  volume={35},
  number={3},
  pages={1758-1768},
  doi={10.1109/TPWRS.2019.2954453}}

@book{fibonacci_search,
  title={Practical Optimization: Algorithms and Engineering Applications},
  author={Antoniou, A. and Lu, W.S.},
  isbn={9780387711072},
  lccn={2007922511},
  year={2007},
  publisher={Springer US}
}

@ARTICLE{idgt_binary2,
  author={Dong, Yingchao and Li, Zhengmao and Zhang, Hongli and Wang, Cong and Zhou, Xiaojun},
  journal={IEEE Trans. Smart Grid}, 
  title={Robust Coordinated Planning of Multi-Region Integrated Energy Systems With Categorized Demand Response}, 
  year={2024},
  volume={15},
  number={6},
  pages={5678-5692}}

@ARTICLE{idgt_binary1,
  author={Lu, Shuai and Lou, Guannan and Gu, Wei and Zhang, Suhan and Yuan, Xiaodong and Li, Qiang and Han, Huachun},
  journal={CSEE J. Power Energy Syst.}, 
  title={Optimal Dispatch for Flexible Uncertainty Sets in Multi-Energy Systems: An IGDT Based Two-Stage Decision Framework}, 
  year={2023},
  volume={9},
  number={6},
  pages={2374-2385},
  doi={10.17775/CSEEJPES.2020.07150}}

@book{RO,
  title={Robust and Adaptive Optimization},
  author={Bertsimas, D. and den Hertog, D.},
  isbn={9781733788526},
  year={2022},
  publisher={Dynamic Ideas LLC}
}

@article{CCG,
title = {Solving two-stage robust optimization problems using a column-and-constraint generation method},
journal = {Oper. Res. Lett.},
volume = {41},
number = {5},
pages = {457-461},
year = {2013},
issn = {0167-6377},
author = {Bo Zeng and Long Zhao}}

@misc{carbon_neutrality,
  author       = {{International Energy Agency (IEA)}},
  title        = {Net Zero by 2050: A Roadmap for the Global Energy Sector},
  year         = {2021},
  howpublished = {\url{https://www.iea.org/reports/net-zero-by-2050}},
  institution  = {International Energy Agency (IEA)}
}

@article{sp1,
title = {Resilience enhancement of distribution network under typhoon disaster based on two-stage stochastic programming},
journal = {Appl. Energy},
volume = {338},
pages = {120892},
year = {2023},
issn = {0306-2619},
author = {Hui Hou and Junyi Tang and Zhiwei Zhang and Zhuo Wang and Ruizeng Wei and Lei Wang and Huan He and Xixiu Wu}
}

@ARTICLE{sp2,
  author={Nazir, Firdous Ul and Pal, Bikash C. and Jabr, Rabih A.},
  journal={IEEE Trans. Power Syst.}, 
  title={A Two-Stage Chance Constrained Volt/Var Control Scheme for Active Distribution Networks With Nodal Power Uncertainties}, 
  year={2019},
  volume={34},
  number={1},
  pages={314-325}}

@ARTICLE{ro1,
  author={Ding, Tao and Liu, Shiyu and Yuan, Wei and Bie, Zhaohong and Zeng, Bo},
  journal={IEEE Trans. Sustain. Energy}, 
  title={A Two-Stage Robust Reactive Power Optimization Considering Uncertain Wind Power Integration in Active Distribution Networks}, 
  year={2016},
  volume={7},
  number={1},
  pages={301-311}}

@ARTICLE{ro2,
  author={Zhang, Cuo and Xu, Yan and Dong, Zhaoyang and Ravishankar, Jayashri},
  journal={IEEE Trans. Smart Grid}, 
  title={Three-Stage Robust Inverter-Based Voltage/Var Control for Distribution Networks With High-Level PV}, 
  year={2019},
  volume={10},
  number={1},
  pages={782-793}}

@ARTICLE{dro1,
  author={Wang, Siyuan and Zhao, Chaoyue and Fan, Lei and Bo, Rui},
  journal={IEEE Trans. Power Syst.}, 
  title={Distributionally Robust Unit Commitment With Flexible Generation Resources Considering Renewable Energy Uncertainty}, 
  year={2022},
  volume={37},
  number={6},
  pages={4179-4190}}

@ARTICLE{dro2,
  author={Rayati, Mohammad and Bozorg, Mokhtar and Cherkaoui, Rachid and Carpita, Mauro},
  journal={IEEE Trans. Smart Grid}, 
  title={Distributionally Robust Chance Constrained Optimization for Providing Flexibility in an Active Distribution Network}, 
  year={2022},
  volume={13},
  number={4},
  pages={2920-2934}}

@book{igdt_Ben_Haim,
  title={Info-gap decision theory: decisions under severe uncertainty},
  author={Ben-Haim, Yakov},
  year={2006},
  publisher={Elsevier}
}

@ARTICLE{w_igdt,
  author={Nasr, Mohamad-Amin and Nasr-Azadani, Ehsan and Nafisi, Hamed and Hosseinian, Seyed Hossein and Siano, Pierluigi},
  journal={IEEE Trans. Industr. Inform.}, 
  title={Assessing the Effectiveness of Weighted Information Gap Decision Theory Integrated With Energy Management Systems for Isolated Microgrids}, 
  year={2020},
  volume={16},
  number={8},
  pages={5286-5299},
  doi={10.1109/TII.2019.2954706}}

@ARTICLE{DDUs,
  author={Sun, Xunhang and Cao, Xiaoyu and Zeng, Bo and Zhai, Qiaozhu and Başar, Tamer and Guan, Xiaohong},
  journal={IEEE Trans. Smart Grid}, 
  title={Stochastic-Robust Planning of Networked Hydrogen-Electrical Microgrids: A Study on Induced Refueling Demand}, 
  year={2025},
  volume={16},
  number={1},
  pages={115-130}}

@article{impact_uncertainty,
  title={Representation of uncertainty in market models for operational planning and forecasting in renewable power systems: a review},
  author={Haugen, Mari and Farahmand, Hossein and Jaehnert, Stefan and Fleten, Stein-Erik},
  journal={Energy Syst.},
  pages={1--36},
  year={2023},
  publisher={Springer}
}

@article{33_bus,
  author={Baran, M.E. and Wu, F.F.},
  journal={IEEE Trans. Power Deliv.}, 
  title={Network reconfiguration in distribution systems for loss reduction and load balancing}, 
  year={1989},
  volume={4},
  number={2},
  pages={1401-1407},
  doi={10.1109/61.25627}}

@book{saa,
	title={Lectures on Stochastic Programming: Modeling and Theory},
	author={Shapiro, Alexander and Dentcheva, Darinka and Ruszczynski, Andrzej},
	edition={3},
	year={2021},
	publisher={SIAM and MOS},
	address={Philadelphia, PA, USA}
}

@ARTICLE{idm_Li_Peng,
  author={Li, Peng and Yang, Ming and Wu, Qiuwei},
  journal={IEEE Trans. Sustain. Energy}, 
  title={Confidence Interval Based Distributionally Robust Real-Time Economic Dispatch Approach Considering Wind Power Accommodation Risk}, 
  year={2021},
  volume={12},
  number={1},
  pages={58-69},
  doi={10.1109/TSTE.2020.2978634}}

\end{document}